\documentclass{LMCS}

\def\dOi{10(1:11)2014}
\lmcsheading%
{\dOi}
{1--30}
{}
{}
{Jun.~17, 2013}
{Feb.~13, 2014}
{}

\ACMCCS{[{\bf Theory of computation}]: Semantics and reasoning---Program
  reasoning---Program verification; [{\bf Software and its
      engineering}]: Software organization and properties---Software
  functional properties---Formal methods---Software verification}

\usepackage{alltt}
\usepackage{amssymb}
\usepackage{cmll}

\usepackage{tikz}
\usetikzlibrary{matrix}
\tikzset{edge/.style={->,font=\scriptsize,line width=.7pt,shorten <= .1cm,shorten >=.1cm}}
\tikzset{over/.style={fill=white,font=\tiny,inner sep=1.5pt}}

\newenvironment{myequation}{%
  \ignorespaces\[\everymath{\displaystyle}\begin{array}{rclr}%
    }{%
  \end{array}\]\ignorespacesafterend%
}

\newcommand\eqdef{\stackrel{\smash{\scriptscriptstyle\mathsf{def}}}{=}}
\newcommand\ZZ{\mathbf{Z}}
\newcommand\NN{\mathbf{N}}
\newcommand\Rule[2]{\frac{\phantom{\big(}\quad{#1}\quad}{\phantom{\Big(}\quad{#2}\quad}
}

\newcommand\Zero{{\mathbf{0}}}

\newcommand\ttt[1]{{\hbox{\rm\texttt{#1}}}}
\newcommand\x{{\ttt{x}}}
\newcommand\y{{\ttt{y}}}
\newcommand\BLANK{\ttt{\char"5F}}

\newcommand{\Depth}{\mathop{\mathrm{depth}}\nolimits}
\newcommand{\nf}{\mathop{\mathrm{nf}}\nolimits}
\newcommand\morphism[1]{\mathrel{\mathop{\longrightarrow}\limits^{#1}}}
\newcommand\PB[2]{{\left\lceil{#2}\right\rceil_{\scriptscriptstyle#1}}}
\newcommand\PC[2]{{{#2}_{\upharpoonright_{#1}}}} %amssymb
\newcommand\PD[2]{{{#2}_{\downharpoonright_{#1}}}} %amssymb

\newcommand\T{\mathcal{T}}

\newcommand\C[1]{\mathtt{#1}}
\newcommand\CC[1]{\mathtt{#1}^{\ttt{\scriptsize-}}}

\newcommand\app[1]{\langle#1\rangle}

\newcommand\subst[1]{\left[\,\SUBST#1;\ENDSUBST\,\right]}
\def\SUBST#1;#2\ENDSUBST{%
  \def\END{#2}%
  \ifx\END\empty%
  #1%
  \else%
  #1\,;\,%
  \SUBST#2\ENDSUBST%
  \fi
}

\renewcommand\leq{\leqslant} %amssymb
\renewcommand\geq{\geqslant} %amssymb
\newcommand\less{\preccurlyeq}  %amssymb
\newcommand\more{\succcurlyeq}  %amssymb
\newcommand\sqless{\sqsubseteq}  %amssymb
\newcommand\red{\to}
\newcommand\comp{\circ}
\newcommand\ccomp{\diamond}

\newcommand\EqDown{\ensuremath{\displaystyle\mathop{\mskip1.5mu\downarrow^{\mskip-13.5mu=}}}}
\newcommand\Down{\ensuremath{\mathop{\downarrow}}}

\newcommand\seq[1]{\overline{#1\,}}
\begin{document}

\title[Size-Change Termination for Constructor Based Languages]
      {The Size-Change Termination Principle\\ for Constructor Based Languages}

\author[P.~Hyvernat]{Pierre Hyvernat}
\address{Laboratoire de Math\'ematiques\\
  CNRS UMR 5126 -- Universit\'e de Savoie\\
  73376 Le Bourget-du-Lac Cedex\\
  France}
\email{pierre.hyvernat@univ-savoie.fr}
\urladdr{\url{http://lama.univ-savoie.fr/~hyvernat/}}
\thanks{This work was partially funded by the French ANR project
  r\'ecr\'e ANR-11-BS02-0010.}


%\date{\today}

\begin{abstract} %%%<<<1

  This paper describes an automatic \emph{termination checker} for a generic
  first-order call-by-value language in ML style. We use the fact that values
  are built from constructors and tuples to keep some information about
  how arguments of recursive calls evolve during evaluation.

  The result is a criterion for termination extending the \emph{size-change
  termination principle} of Lee, Jones and Ben-Amram that can detect size
  changes inside subvalues of arguments. Moreover the corresponding algorithm
  is easy to implement, making it a good candidate for experimentation.
\end{abstract} %%%>>>1

\maketitle

%\tableofcontents

\section*{Introduction} %%%<<<1

Our goal is to automatically check the termination of mutually recursive
definitions written in a first-order call-by-value language in ML style. The
problem is of course undecidable and we can only hope to capture \emph{some}
recursive definitions.
Lee, Jones and Ben-Amram's size-change termination principle (SCT) is a
simple, yet surprisingly strong sufficient condition for termination of
programs~\cite{SCT}.  It relies on a notion of \emph{size} of values and a
static analysis interpreting a recursive program as a \emph{control-flow
graph} with information about how the size of arguments evolves during
recursive calls. The procedure checking that such a graph is ``terminating''
amounts to a (conceptually) simple construction of a graph of paths.

We specialize and extend this principle to an ML-like language where
first-order values have a specific shape: they are built with $n$-tuples and
constructors. It is then possible to record more information about
arguments of recursive calls than ``decreases strictly'' or ``decreases''. The
main requirement is that the set of possible informations is finite, which we
get by choosing bounds for the \emph{depth} and the \emph{weight} of the terms
describing this information. We obtain a parametrized criterion for checking
termination of first-order recursive programs. The weakest version of this
criterion corresponds to the original SCT where the size of a value is its
depth.
An important point is that because we know some of the constructors present in
the arguments, it is possible to ignore some paths in the control-flow graph
because they cannot correspond to real evaluation steps. Moreover, it makes it
possible to inspect \emph{subvalues} of the arguments and detect a ``local''
size change. Another important point is that there is a simple syntax directed
static analysis that can be done in linear time.

The criterion has been implemented as part of the PML~\cite{EJC} language,
where it plays a central role: PML has a notion of \emph{proofs}, which are
special programs that need to terminate. As far as usability is concerned,
this criterion was a success: it is strong enough for our purpose, its output
is usually easy to predict and its implementation was rather straightforward.
The core consists of about 600 lines of OCaml code without external
dependencies.\footnote{A standalone version is available from
\url{http://lama.univ-savoie.fr/~hyvernat/research.php}}
%Readers can test the termination checker with the PML language
%(\cite{PMLwww}). The code is available in the PML
%distribution\footnote{\url{http://lama.univ-savoie.fr/~pml/}} or as a single
%file on the author's webpage for easy perusal.

\smallbreak
The paper is organized as follows: after introducing the ambient programming
language and some paradigmatic examples, we first define an abstract
interpretation for \emph{calls} and look at their properties. This makes it
possible to give an abstract interpretation for sets of recursive definitions
as {control-flow graphs}. A subtle issue arises when we try to make the set of
possible interpretations finite, making the notion of composition not
associative in general. We then describe and prove the actual criterion. We
finish with an appendix giving some technical lemmas, details about the
implementation and a simple static analysis.

\subsubsection*{Comparison with other work}

%Using the size-change termination principle for a language with tuples and
%constructors was already done by D. Wahlstedt \cite{David}, who used it
%in the context of Martin-L\"of theory of dependent types. However, SCT was
%there used as an element of the proof that the theory considered was
%well-founded. In particular, it is not possible to write looping programs in
%this system. The Agda programming language contains both Martin L\"of
%dependent type theory and the ability to write unrestricted recursive
%functions.

Two aspects of this new criterion appeared in the
literature~\cite{deltaSCT,calling_contexts} (see
Section~\ref{sub:comparison}), but what seems to be new here is that the
algorithm for testing termination is, like for the original SCT, ``finitary''.
Once the static analysis is done ---and this can be as simple as a linear-time
syntactical analysis of the definitions--- one needs only to compute the graph
of paths of the control-flow graph and inspect its loops. This makes it
particularly easy to implement from scratch as it needs not to rely on
external automatic proof-checker~\cite{calling_contexts} or integer linear
programming libraries~\cite{deltaSCT}.

One advantage of this minimalistic approach is that a formal proof of the
criterion is probably easier, making the criterion well-suited for proof
assistants based on type theory like Coq~\cite{Coq} or~Agda~\cite{Agda}. The
closest existing criterion seems to be the termination checker of Agda. It is
based on the ``\texttt{foetus}'' termination checker~\cite{foetus}. The
implementation incorporates a part of SCT but unfortunately, the exact
criterion isn't formally described anywhere.

It should be noted that native datatypes (integers with arithmetic operations
for example) are not addressed in this paper. This is not a problem as proof
assistants don't directly use native types. Complementing the present approach
with such internal datatypes and analysis of higher-order
programs~\cite{Sereni05terminationanalysis} is the subject of future research.

\subsection*{Ambient Programming Language} %%%<<<2
\label{sub:ambient_language}

The programming language we are considering is a first-order call-by-value
language in ML-style. It has constructors, pattern-matching, tuples and
projections.
The language is described briefly in Figure~\ref{fig:language} and the syntax
should be obvious to anyone familiar with an ML-style language.
The~``\ttt{match}'' construction allows to do pattern matching,
while~``$\pi_i$'' is used for projecting a tuple on one of its components. The
only proviso is that all constructors are unary and written
as~``$\C{C}\ttt{[}u\ttt{]}$''. Note that the~\ttt{f} in the grammar for expressions can
either be one of the functions that are being inductively defined, or any
function in the global environment.
%We add a non-deterministic choice operator~\ttt{flip}, but this doesn't play
%an important role in our analysis.
Other features like~\ttt{let} expressions, exceptions, (sub)typing etc. can
easily be added as they don't interfere with the criterion. (They might make
the static analysis harder though.)

\begin{figure}
%\figrule
%\vspace{-\baselineskip}
\begin{myequation}
  \textit{program}
  &::=&
  \ttt{val rec} \quad \textit{def} \quad (\ttt{and} \quad \textit{def\/})^*\\

  \textit{def}
  &::=&
  \ttt{f x\(\sb1\) x\(\sb2\) } \dots \ttt{ x\(\sb n\) = }\textit{term}\\

  \textit{expr}
  &::=&
  \ttt{x\(\sb k\)}                                                       \quad|\quad
%  \textit{cste}                                                 \quad|\quad
  \ttt{f}                                                       \quad|\quad
%  \ttt{flip}\ \textit{expr}\ \textit{expr}                      \quad|\quad
  \textit{expr}\ \textit{expr}^+                                \quad|\quad \\
&&\ttt{C[{\it expr}]}                                           \quad|\quad
  \ttt{(}\textit{expr}, \dots, \textit{expr}\ttt{)}             \quad|\quad \\
&& \ttt{match}\ \textit{expr}\ \ttt{with}\ \textit{branch}^+    \quad|\quad
  \pi_i\,\textit{expr} \quad{\scriptstyle\textrm{(with~$i>0$)}}\\
  \textit{branch}
  &::=&
  \ttt{| C[x\(\sb k\)] ->}\ \textit{expr}
  \quad|\quad \ttt{| \_ ->}\ \textit{expr}

\end{myequation}
\caption{syntax of the programming language}
\label{fig:language}
\end{figure}

The operational semantics is the usual one and we only consider programs whose
semantics is well defined\label{rk:safe}. This can be achieved using
traditional Hindley-Milner type checking / type inference~\cite{Milner} or a
\emph{constraint checking} algorithm~\cite{EJC} ensuring that
\begin{itemize}
  \item a constructor is never projected,
  \item a tuple is never matched,
  \item an $n$-tuple is only projected on its~$i$-th component if~$1\le i\le
  n$.
\end{itemize}
To simplify the presentation, we assume that functions have an arity and are
always fully applied. Moreover, we suppose that the arguments of functions are
all first-order values. These constraints are relaxed in the actual
implementation.

%Those constraint can be checked at the same time as the safety of the
%definitions.

\smallbreak
An important property of this language is that non-termination can only be the
result of evaluation going through an infinite sequence of calls to recursive
functions~\cite{EJC}. A consequence of that is that it is not possible to use
the notions described in this paper directly for languages where a fixed point
combinator can be defined without recursion. Extensions similar to the work of
Jones and Bohr for untyped languages~\cite{JonesB08} might be possible, at the
cost of a greatly increased complexity of implementation.

\smallbreak
A \emph{first-order value} is a closed expression built only with constructors and
(possibly empty) tuples. Examples include unary natural numbers built with
constructors~``\ttt{Z}'' and~``\ttt{S}'' or lists built with
constructors~``\ttt{Nil}'' and~``\ttt{Cons}''.
The \emph{depth} of a value is
\begin{myequation}
  \Depth(\ttt{C[$u$]})           &\eqdef& 1 + \Depth(u)\\
  \Depth\big((u_1,\dots,u_n)\big)   &\eqdef& \max_{1\leq i\leq n} \big(1+\Depth(u_i)\big)
  \ \hbox{.}
\end{myequation}
Note that values are not explicitly typed and that depth counts all
constructors. For example, the depth of a list of natural numbers counts
the~\ttt{Nil}, \ttt{Cons}, \ttt{S} and~\ttt{Z} constructors, as well as the
tuples coming with the~$\C{Cons}$ constructors.

To make examples easier to read, we will deviate from the grammar of
Figure~\ref{fig:language} and use ML-like deep pattern-matching, including
pattern-matching on tuples. Moreover, parenthesis around tuples will be
omitted when they are the argument of constructors. For example, here is
how we write the usual~\ttt{map} function: \label{map_list}
{\small\begin{alltt}
  val rec map x = match x with Nil[]  ->  Nil[]
                             | Cons[a,y]  ->  Cons[f a, map y]
\end{alltt}}\noindent
Without the previous conventions, the definition would look like
{\small\begin{alltt}
  val rec map x = match x with | Nil[y]  ->  Nil[()]
                               | Cons[y]  ->  Cons[(f \(\pi\sb1\)y, map \(\pi\sb2\)y)]
\end{alltt}}\noindent
Note that because we here restrict to first-order arguments, we cannot
formally make~\ttt{f} an argument of~\ttt{map}. We thus assume that it is
a predefined function. This constraint is relaxed in the actual
implementation.

%%%>>>2

\subsubsection*{Vocabulary and notation} %%%<<<2

We use a fixed-width font, possibly with subscripts, for syntactical
tokens:~``$\x$'',~``$\y$'' or~``$\x_i$'' for variables, ``\ttt{f}''
or~``\ttt{g}'' for function names, ``$\C{A}$'' for a constructor, etc. The
only exception will be the letter~$\pi$, used to represent a projection. Meta
variables representing terms will be written with italics:~``$t$'',~``$u$''
or~``$t_i$'' etc.

\medbreak
For a set of mutual recursive definitions
{\small\begin{alltt}
    val rec f \(\x\sb1\) \(\x\sb2\) \(\x\sb3\) = ... g \(t\sb1\) \(t\sb2\) ...
  and     g \(\y\sb1 \)\(\y\sb2 \)   = ...
\end{alltt}}\noindent
where~$\x_1$, $\x_2$, $\x_3$, $\y_1$ and~$\y_2$ are variables and~$t_1$
and~$t_2$ are expressions,
\begin{itemize}
\item ``\ttt{\(\x\sb1\)}'', ``\ttt{\(\x\sb2\)}'' and ``\ttt{\(\x\sb3\)}'' are the
\emph{parameters of the definition of~\ttt{f}},
\item ``\ttt{g \(t\sb1\) \(t\sb2\)}'' is a \emph{call site from~\ttt{f} to~\ttt{g}},
\item ``$t_1$'' and ``$t_2$'' are the \emph{arguments
of~\ttt{g} at this call site}.
\end{itemize}
We usually abbreviate those to \emph{parameters}, \emph{call} and
\emph{arguments}.

%%%>>>2

\subsection*{Examples} %%%<<<2

Here are some examples of recursive (ad-hoc) definitions that are accepted by our
criterion.

\begin{itemize}
  \item All the structurally decreasing inductive functions, like the
    \ttt{map} function given previously are accepted.

  \item Our criterion generalizes the original SCT (where the
    size of a value is its depth), and thus, all the original
    examples~\cite{SCT} pass the test. For example, the Ackermann function is
    accepted:
{\small\begin{alltt}
  val rec ack \(\x\sb1\) \(\x\sb2\) = match (\(\x\sb1\),\(\x\sb2\)) with
                          (Z[],Z[]) -> S[Z[]]
                        | (Z[],S[n]) -> S[S[n]]
                        | (S[m],Z[]) -> ack m S[Z[]]
                        | (S[m],S[n]) -> ack m (ack S[m] n)
\end{alltt}}\label{def:ackermann}%

  \item In the original SCT, the size information is lost as soon as a value
    increases. We do support a local bounded increase of size as in
{\small\begin{alltt}
    val rec \(\ttt{f}\sb1\) x = \(\ttt{g}\sb1\) A[x]
      and \(\ttt{g}\sb1\) x = match x with A[A[x]] -> \(\ttt{f}\sb1\) x
                            | _    -> ()
\end{alltt}}\label{def:f1g1}\noindent
    The call from~\ttt{\(\ttt{f}\sb1\)} to~\ttt{\(\ttt{g}\sb1\)} (that
    increases the depth by~$1$) is harmless because it is followed by a call
    from~\ttt{\(\ttt{g}\sb1\)} to~\ttt{\(\ttt{f}\sb1\)} (that decreases the
    depth by~$2$).

  \item In the definition
%{\small\begin{alltt}
%  val rec \(\ttt{f}\sb2\) x = match x with    A[x] -> \(\ttt{f}\sb2\) B[x]
%                               | B[x] -> \(\ttt{f}\sb2\) x
%                               | _    -> ()
%\end{alltt}}\noindent
{\small\begin{alltt}
val rec \(\ttt{f}\sb2\) x = match x with   A[x] -> \(\ttt{f}\sb2\) B[C[x]]
                            | B[x] -> \(\ttt{f}\sb2\) x
                            | C[x] -> \(\ttt{f}\sb2\) x
\end{alltt}}\noindent
    the size of the argument increases in the first recursive call. This alone
    would make the definition non size-change terminating for the original SCT. However,
    the constructors and pattern matching imply that the first recursive
    call is necessarily followed by the second and third one, where the size
    decreases at last. This function passes the improved test.

  \item In the definition
{\small\begin{alltt}
  val rec push_left x =
    match x with Leaf[] -> Leaf[]
               | Node[t, Leaf[]] -> Node[t, Leaf[]]
               | Node[\(\ttt{t}\sb1\), Node[\(\ttt{t}\sb2\),\(\ttt{t}\sb3\)]] -> push_left Node[Node[\(\ttt{t}\sb1\),\(\ttt{t}\sb2\)],\(\ttt{t}\sb3\)]
\end{alltt}}\noindent\label{ex:push_left}%
    the depth of the argument does not decrease but the depth of its right
    subtree does. In the original SCT, the user could choose the ad-hoc notion
    of size ``depth of the right-subtree''. Our criterion will see that this
    is terminating without help.

\end{itemize}

%%%>>>2

\subsubsection*{Idea of the Algorithm} %%%<<<2

Just like the original SCT, our algorithm works by making an abstract
interpretation of the recursive definitions as a control-flow graph. This is
done by a static analysis independent of the actual criterion. A simple,
syntactical static analysis that allows to deal with the examples of the paper
is described in Appendix~\ref{sec:static_analysis}.
This control-flow graph only represents the evolution of arguments of recursive
calls. For example both the~\ttt{map} function and
the~\ttt{last} function
{\small\begin{alltt}
  val rec last x = match x with Cons[a,Nil[]] -> a
                              | Cons[_,x] -> last x
\end{alltt}}\noindent
have the same control-flow graph: when the function is called on a non-empty
list, it makes a recursive call to the tail of the list.

Ideally, each argument to a call should be represented by a transformation
describing how the argument is obtained from the parameters of the defined
function. To make the problem tractable, we restrict to transformations
described by a simple term language. For example, the argument of
the~\ttt{map}/\ttt{last} functions is described by~``$\pi_2 \CC{Cons} \x$'':
starting from parameter~$\x$, we remove a~$\C{Cons}$ and take the second
component of the resulting tuple. When a function has more than one parameter,
each argument of the called function is described by a term with free
variables among the parameter of the calling function.

Checking termination is done by finding a sufficient condition for the
following property of the control-flow graph: \emph{no infinite path of the
graph may come from an infinite sequence of real calls}. The two main reasons
for a path to not come from a sequence of real calls are:
\begin{itemize}
  \item there is an incompatibility in the path: for example, it is not
    possible to remove a~$\C{Cons}$ from the~$\C{Nil}$ value,
  \item it would make the depth of some value negative: for example,
    it is not possible to remove infinitely many~$\C{Cons}$ from
    a given list.
\end{itemize}
In order to do that, we will identify loops that every infinite path must go
through, and check that for all these ``coherent'' loops, there is some part
of an argument that decreases strictly.
For example, in the definition of~\ttt{push\_left}
(page~\pageref{ex:push_left}), the right subtree of the argument is
decreasing, which makes the function pass the termination test.

%%%>>>2
%%%>>>1

\section{Interpreting Calls} %%%<<<1

\subsection{Terms and Reduction} %%%<<<2

The next definition gives a way to describe how an argument of a recursive
calls is obtained from the parameters of the calling function:
\begin{defi}
  Representations for arguments are defined by the following grammar
\[
  t \in \T \quad::=\quad
    \x_k \mskip 15mu|\mskip 15mu
    \underbrace{\C{C}t \mskip 15mu|\mskip 15mu
    (t_1,\dots,t_n)}_{\hbox{\scriptsize constructors}} \mskip 15mu|\mskip 15mu
    \underbrace{\CC{C}t \mskip 15mu|\mskip 15mu
    \pi_i t}_{\hbox{\scriptsize destructors}} \mskip 15mu|\mskip 15mu
    t_1 + t_2 \mskip 15mu|\mskip 15mu
    \Zero \mskip 15mu|\mskip 15mu
    \app{w} t
\]
where~$\x_k$ can be any variable of the ambient language,~$n\geq 0$, $i\geq1$
and~$w\in\ZZ_\infty=\ZZ\cup\{\infty\}$.
We write~$\T(\x_1,\dots,\x_n)$ for the set of terms whose variables are
in~$\{\x_1,\dots,\x_n\}$.

We enforce linearity (or $n$-linearity for~$n$-tuples) for
all term formation operations with the following equations:
\[\begin{array}{rclcrcl}
    \C{C} \Zero &=& \Zero &&
    \C{C}(t_1+t_2) &=& \C{C}t_1 + \C{C}t_2  \\
    (\dots,\Zero,\dots) &=& \Zero &&
    (\dots,t_1+t_2,\dots) &=& (\dots,t_1,\dots) + (\dots,t_2,\dots) \\
    \CC{C} \Zero &=& \Zero &&
    \CC{C}(t_1+t_2) &=& \CC{C}t_1 + \CC{C}t_2  \\
    \pi_i \Zero &=& \Zero &&
    \pi_i(t_1+t_2) &=& \pi_i t_1 + \pi_i t_2  \\
    \app{w} \Zero &=& \Zero &&
    \app{w} (t_1+t_2) &=& \app{w} t_1 + \app{w} t_2\ .  \\
\end{array}\]
We also quotient~$\T$ by associativity, commutativity, neutrality of~$\Zero$,
\emph{and idempotence} of~$+$:
\[\begin{array}{rclcrcl}
    t_1 + (t_2 + t_3) &=& (t_1 + t_2) + t_3  &&
  t + \Zero &=& t  \\
    t_1 + t_2 &=& t_2 + t_1  &&
  t + t &=& t \ . \\
\end{array}\]
\end{defi}\noindent
The intuition is that:
\begin{itemize}
  \item $\x_k$ is a parameter of the calling function.
  \item $\C{C}$ is a constructor and~$(\BLANK,\dots,\BLANK)$ is a
    tuple.
  \item $\pi_i$ is a projection. It gives access to the~$i$th component of a
    tuple.
  \item $\CC{C}$ corresponds to a branch of pattern matching. It removes
    the~$\C{C}$ from a value.
  \item $\Zero$ is an artifact used to represent an error during
    evaluation. Since we only look at well defined programs (see remark on
    page~\pageref{rk:safe}), any~$\Zero$ that
    appears during analysis can be ignored as it cannot come from an actual
    computation.
  \item $t_1+t_2$ acts as a non-deterministic choice.
    %Either~$t_1$ or~$t_2$ is going to be chosen.
    Those sums will play a central role in our analysis of control-flow graphs.
    %when analysing definitions even when the initial static analysis doesn't
    %introduce them.
  \item $\app{w}$ stands for an unknown term that may increase the depth of
    its argument by at most~$w$. For example, if~$w<0$, then the depth
    of~$\app{w}v$ is strictly less than the depth of~$v$.
    Those terms will serve as \emph{approximations} of other terms:
    for example, both~$\C{C}t$ and~$\C{D}t$ can be approximated by~$\app{1}t$,
    but each one contains strictly more information than~$\app{0} t$.
\end{itemize}
%For example, there is a single recursive call in the definition of \ttt{map}:
%{\small\begin{alltt}
%  val rec map l = match l with    Nil[]  ->  Nil[]
%                                | Cons[a,l]  ->  Cons[f a, \underline{map l}]
%\end{alltt}}\noindent
%Its first argument is simply the parameter~\ttt{f} of the calling function,
%and its second argument is obtained from the parameter by removing
%a~$\C{Cons}$ and projecting:~$\CC{Cons}\pi_2\ttt{l}$. This call
%will be represented by the
%substitution~$\subst{\ttt{l}:=\CC{Cons}\pi_2\ttt{l}}$.
%
There is a natural notion of reduction on terms:
\begin{defi}\label{def:reduction}
  We define a reduction relation on~$\T$:
  \[\begin{array}{lrclcrcl}
(1) & \CC{C}\C{C} t & \red & t &\mskip20mu&
      \pi_i(t_1,\dots,t_n) & \red & t_i \hbox{\quad if $1\leq i\leq n$}\\
      \noalign{\medbreak}
(2) & \app{w} \C{C} t & \red & \app{w+1} t &&
      \app{w} (t_1,\dots,t_n) & \red & \sum_{1\leq i\leq n}\app{w+1} t_i
      \hbox{\quad if $n>0$}\\
(2) & \CC{C} \app{w} t & \red & \app{w-1} t &&
      \pi_i \app{w} t & \red & \app{w-1} t \\
(2) &  \app{w} \app{v} t & \red & \app{w+v} t \\
      \noalign{\medbreak}
(3) & \pi_i\C{C}t & \red & \Zero &&
      \pi_i(t_1,\dots,t_n) & \red & \Zero \hbox{\quad if $i>n$}\\
(3) & \CC{C}(t_1,\dots,t_n) & \red & \Zero && \CC{C}\C{D} t & \red & \Zero
      \hbox{\quad if $\C{C}\neq\C{D}$}
  \end{array}\]
  The symbol~``$+$'' for elements of~$\ZZ_\infty$ denotes the obvious
  addition, with~$\infty+\infty = \infty+n = \infty$.
\end{defi}
This reduction extends the operational semantics of the ambient language: the
two rules from group~$(1)$ correspond to the evaluation mechanism and the four
rules from group~$(3)$ correspond to unreachable states of the evaluation
machine. The five rules from group~$(2)$ explain how approximations behave.
Note in particular that:
\begin{itemize}
  \item a~$\app{w}$ absorbs constructors on its right and destructors on its
    left,
  \item a~$\app{w}$ may approximate some projections and we don't know which
    components of a tuple it may access. This is why a sum appears in
    the reduction.
\end{itemize}

\begin{lem}\label{lem:SN}
  The reduction~$\red$ is strongly normalizing and confluent. We
  write~$\nf(t)$ for the unique normal form of~$t$.
\end{lem}
We write~$t\approx u$ when~$t$ and~$u$ have the same normal form. This
lemma implies that~$\approx$ is the least equivalence relation containing
reduction.
\begin{proof}[Proof of Lemma~\ref{lem:SN}]
  Strong normalization is easy as the depth of terms decreases strictly during
  reduction. By Newman's lemma, confluence thus follows from local confluence
  which follows from examination of the critical pairs:
  \[\begin{array}{c@{\qquad}c@{\qquad}c}
      \CC{D}\app{w} \C{C} t &
      \CC{D}\app{w} (t_1,\dots,t_n) &
      \CC{C}\app{w}\app{v} t \\
      \pi_i\app{w} \C{C} t &
      \pi_i\app{w} (t_1,\dots,t_n) &
      \pi_1\app{w}\app{v} t \\
      \app{w}\app{v}\C{C} t &
      \app{w}\app{v}(t_1,\dots,t_n) &
      \app{w}\app{v}\app{u} t
      \ .
  \end{array}\]
  For example,~$\CC{D}\app{w} \C{C} t$ reduces both to~$\app{w-1}\C{C}t$
  and~$\CC{D}\app{w+1} t$. Luckily, those two terms reduce
  to~$\app{w} t$. The same holds for the eight remaining critical pairs.
\end{proof}
Call a term~$t\in\T$ \emph{simple} if it is in normal form and doesn't
contain~$+$ or~$\Zero$. We have:
\begin{lem}
  Every term~$t\in\T$ reduces to a (possibly empty) sum of simple terms,
  where the empty sum is identified with~$\Zero$.
\end{lem}
\begin{proof}
  This follows from the fact that all term constructions are linear and that
  the reduction is strongly normalizing. Note that because of confluence,
  associativity, commutativity and idempotence of~$+$, this representation is
  essentially unique.
\end{proof}
Simple terms have a very constrained form: all the constructors are on the
left and all the destructors are on the right. More precisely:
\begin{lem}\label{lem:normal_form}
  The simple terms of~$\T$ are generated by the grammar
  \begin{myequation}
    t
    &::=&
    \C{C}t              \mskip 15mu|\mskip 15mu
    (t_1,\dots,t_n)     \mskip 15mu|\mskip 15mu
    \seq d              \mskip 15mu|\mskip 15mu
  \app{w}\seq d\qquad\hbox{\small$(n>0)$}\\
    \seq d
    &::=&
      () \mskip 15mu|\mskip15mu \seq{d_v}\\
    \seq{d_v}
    &::=&
    \x_k            \mskip 15mu|\mskip 15mu
    \pi_i \seq{d_v}    \mskip 15mu|\mskip 15mu
    \CC{C} \seq{d_v}
  \end{myequation}
  We will sometimes write~$\seq d\x = d_1\cdots d_n\x$ for some~$\seq{d_v}$ ending with
  variable~$\x$.

  The length~$|\seq d|$ of~$\seq d$ is the number of
  destructors~$\CC{C}$/$\pi_i$ it contains.
\end{lem}
We now introduce a preorder describing approximation.
\begin{defi}\label{def:preorder}
  The relation~$\less$ is the least preorder on~$\T$
  satisfying
  \begin{itemize}
    \item $\less$ is contextual: if $t$ is a term, and
      if~$u_1\less u_2$, then~$t[\x:=u_1]\less t[\x:=u_2]$,
    \item $\less$ is compatible with~$\approx$: if $t\approx u$ then~$u \less t$
      and~$t\less u$,
    \item $\less$ is compatible with~$+$ and~$\Zero$ : $\Zero\less t$ and~$t
      \less t+u$,
    \item if~$v \leq w$ in~$\ZZ_\infty$ then~$\app{v} t \less \app{w} t$,
    \item $t \less \app{0} t$.
  \end{itemize}
  When~$t\less u$, we say that ``$t$ is finer than~$u$'' or that ``$u$ is an
  approximation of~$t$''.
\end{defi}
This definition implies for example that~$\C{C}\x \less \app{0}\C{C}\x \approx
\app{1}\x$, and thus, by contextuality, that~$\C{C}t \less \app{1}t$ for any
term~$t$. Appendix~\ref{sec:inductive_approximation} gives a characterization
of this preorder that is easier to implement because it doesn't use
contextuality. It implies in
particular the following lemma:
\begin{lem}\label{lem:approximation_destructors}
  We have $\app{w}\seq d \less \app{w'}\seq b$ if and only~$\seq b$ is a
  suffix of~$\seq d$ and~$w+|\seq b| \leq w'+|\seq b|$. In
  particular,~$\app{w}() \less \app{w'}()$ if and only if~$w\leq w'$.
\end{lem}
An important property is that a finer term has at least as many head constructors
as a coarser one:
\begin{lem}\label{lem:orderHeadConstructors}
  If~$u\not\approx\Zero$, we have:
  \[
    u \less \C{C}v \quad\iff\quad u=\C{C}{u'} \hbox{ with } u' \less v
  \]
  and
  \[
    u \less (v_1,\dots,v_n) \quad\iff\quad u=(u_1, \dots,u_n) \hbox{ with }
    u_1 \less v_1 \dots u_n \less v_n
  \]
\end{lem}
\begin{proof}
  There is a simple, direct inductive proof, but this lemma also follows from
  the characterization of~$\less$ in
  Appendix~\ref{sec:inductive_approximation}
  (Lemma~\ref{lem:inductive_approximation}).
\end{proof}
The next lemma gives some facts about the preorder~$(\T,\less)$ that may help
getting some intuitions.
\begin{lem}\label{lem:app_preorder}
  We have
  \begin{itemize}
    \item $\Zero$ is the least element,
    \item $\app{\infty}()$ is the greatest element of~$\T()$, the set of
    \emph{closed terms},
    \item $+$ is a least-upper bound, i.e., $t_1 + t_2 \less u$ iff~$t_1\less u$
      and~$t_2\less u$,
    \item if~$t$ and~$u$ are simple, then~$t\less u \hbox{ and } u\less t$
      iff~$t=u$.
%    \item a term~$t$ is an atom, in the order theory sense, iff it reduces to
%    a simple term without any~$\app{w}$.
  \end{itemize}
\end{lem}
The last point follows from Lemmas~\ref{lem:approximation_destructors}
and~\ref{lem:orderHeadConstructors}. The rest is direct.
%%%>>>2

\subsection{Substitutions and Control-Flow Graphs} %%%<<<2

Just like a term is meant to represent one argument of a recursive call,
a substitution~$\subst{\x_1:=u_1;\dots;\x_n:=u_n}$ is meant to represent
\emph{all} the arguments of a recursive call to an~$n$-ary function. In order
to follow the evolution of arguments along several recursive calls, we need to
\emph{compose} substitutions: given some terms~$t$,~$u_1$, \dots,$u_n$
in~$\T$, we define~$t\subst{\x_1:=u_1;\dots;\x_n:=u_n}$ as the parallel
substitution of each~$\x_i$ by~$u_i$. The composition~$\tau\comp\sigma$ of two
substitutions~$\tau=\subst{\x_1:=u_1;\dots;\x_n=u_n}$ and~$\sigma$ is simply the
substitution~$\tau\comp\sigma = \subst{\x_1:=u_1\sigma;\dots;\x_n:=u_n\sigma}$.

\begin{lem}
  Composition of substitutions is associative and monotonic (for the pointwise
  order) on the right and on the left: if~$\tau_1 \less \tau_2$ then~$\sigma
  \comp \tau_1 \less \sigma \comp \tau_2$ and~$\tau_1 \comp \sigma \less
  \tau_2 \comp \sigma$.
\end{lem}
\proof
  Associativity is obvious. Monotonicity on the left follows from the fact
  that~$\less$ is contextual.
  For monotonicity on the right, we show~``$t_1\less t_2$
  implies~$t_1[\x:=u]\less t_2[\x:=u]$'' by induction on~$t_1\less t_2$.
  The only interesting case is when~$t_1 \less t_2$ because~$t_1 = t[\y:=v_1]$
  and~$t_2=t[\y:=v_2]$ and~$v_1\less v_2$. By induction hypothesis, we
  have~$v_1[\x:=u] \less v_2[\x:=u]$. There are two cases:
  \begin{itemize}
    \item if~$\x=\y$, we have~$t_i[\x:=u] = t[\x:=v_i][\x:=u] =
      t\big[\x:=v_i[\x:=v_i]\big]$ and we get~$t_1[\x:=u] \less t_2[\x:=u]$ by
      contextuality applied to the induction hypothesis;

    \item if~$\x\neq\y$, we have~$t_i[\x:=u] = t[\y:=v_i][\x:=u] =
      t[\x:=u]\big[\y:=v_i[\x:=u]\big]$, and here again, we get~$t_1[\x:=u]
      \less t_2[\x:=u]$ by contextuality applied to the induction hypothesis.\qed
  \end{itemize}

We can now define what the abstract interpretations for
our programs will be:
\begin{defi}
  A \emph{control-flow graph} for a set of mutually recursive definitions is a labeled
  graph where:
  \begin{itemize}
    \item vertices are function names,
%    \item each call site from~\ttt{f} to~\ttt{g} corresponds exactly to one arc
%      from ``\ttt{f}'' to ``\ttt{g}'',
    \item if the parameters of~\ttt{f} are~$\y_1, \dots,\y_m$ and the
      parameters of~\ttt{g} are~$\x_1,\dots,\x_n$, the labels of arcs
      from~\ttt{f} to~\ttt{g} are substitutions~$\subst{\x_1 := u_1; \dots;
      \x_n := u_n}$, where each~$u_i$ is a term in~$\T(\y_1,\dots,\y_m)$.
  \end{itemize}
\end{defi}

\noindent That a control-flow graph is \emph{safe} (Definition~\ref{d:safety}) means
that it gives approximations of the real evolution of arguments of the
recursive calls during evaluation. Since we are in a call-by-value language,
those arguments are first-order values of
the ambient language (see page~\pageref{sub:ambient_language}) which can be
embedded in~$\T()$.

\begin{defi}
  A \emph{value} is a simple term of~$\T()$, i.e., a simple closed term. An
  \emph{exact value} is a value which doesn't contain any~$\app{w}$,
  with~$w\in\ZZ_\infty$.
\end{defi}
First-order values of the ambient language correspond precisely to
exact values in~$\T$.
We can now define safety formally:
\begin{defi}\label{d:safety}
  Let~$G$ be a control-flow graph for some recursive definitions,
  \begin{enumerate}
  \item suppose we have a call site from~\ttt{f} to~\ttt{g}:
  {\small\begin{alltt}
  val rec f \(\x\sb1\) \(\x\sb2\) ... \(\x\sb{n}\) =
    ... g \(u\sb1\) ... \(u\sb{m}\)
    ...\end{alltt}}\noindent
    An arc~$\ttt{f}\morphism{\sigma}\ttt{g}$ in~$G$ \emph{safely represents}
  this particular call site if for every substitution~$\rho$ of the parameters by
  \emph{exact values}~$\rho=\subst{\x_1:=v_1;\dots;\x_n:=v_n}$,
  we have
  \[
    \subst{
      \y_1:=[\![u_1]\!]_\rho
      ;\dots;
      \y_m:=[\![u_m]\!]_\rho
    }
    \quad\less\quad
    \sigma\comp\rho
  \]
  where each~$[\![u_i]\!]_\rho$ is the value of~$u_i$ given by the operational
  semantics of the language, in the environment where each variable~$\x_i$ has
  value~$\rho(\x_i)$.

  \item A set of mutually recursive definitions is \emph{safely represented}
    by a control-flow graph if each call site is safely represented by at
    least an arc in the graph.

  \end{enumerate}
\end{defi}
For example, the recursive definitions for~\ttt{\(\ttt{f}\sb1\)}
and~\ttt{\(\ttt{g}\sb1\)} from page~\pageref{def:f1g1} and the Ackermann
function are safely represented by the following control-flow graphs:
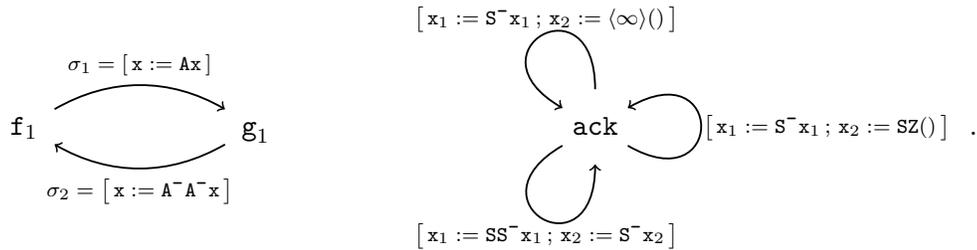
\[
  \begin{tikzpicture}[baseline=(m-1-1.base)]
    \matrix (m) [matrix of math nodes, row sep={3em,between origins},
    column sep={8em,between origins},text depth=.25ex, text height=1.5ex]
      { \ttt{\(\ttt{f}\sb1\)} & \ttt{\(\ttt{g}\sb1\)} \\ };
    \path[edge]
      (m-1-1) edge[bend left=30] node[auto] {$\sigma_1=\subst{\x := \C{A} \x}$} (m-1-2)
      (m-1-2) edge[bend left=30] node[auto] {$\sigma_2=\subst{\x := \CC{A}\CC{A}\x}$} (m-1-1);
  \end{tikzpicture}
  \qquad
  \qquad
  \begin{tikzpicture}[baseline=(m-1-1.base)]
    \matrix (m) [matrix of math nodes, row sep={3em,between origins},
    column sep={6em,between origins},text depth=.25ex, text height=1.5ex]
      { \ttt{ack} \\ };
      \path[edge] (m-1-1.center) edge[loop, distance=2cm, out=-30,in=30, shorten <= .5cm,shorten >=.5cm]
                                 node[right] {$\subst{\x_1:=\CC{S} \x_1;\x_2:=\C{S}\C{Z}()}$}
                (m-1-1.center);
    \path[edge] (m-1-1.center) edge[loop, distance=2cm, out=90,in=150, shorten <= .5cm,shorten >=.5cm]
                               node[above] {$\subst{\x_1:=\CC{S} \x_1;\x_2:=\app{\infty}()}$}
                (m-1-1.center);
    \path[edge] (m-1-1.center) edge[loop, distance=2cm, out=-150,in=-90,shorten <= .5cm,shorten >=.5cm]
                               node[below] {$\subst{\x_1:=\C{S}\CC{S} \x_1;\x_2:=\CC{S} \x_2}$}
                (m-1-1.center);
  \end{tikzpicture}
  \ .
\]\label{graph:f1g1}
The arc~$\sigma_2$ safely represents the call 
{\small\begin{alltt}
    ... match x with A[A[y]] -> \(\ttt{f}\sb1\) y
\end{alltt}}\noindent
because any value~$v$ of~$\x$ reaching the call must be of the
form~$\C{A}\C{A}v'$. We thus have~$[\![\x]\!] = v$ and~$[\![\y]\!]=v'$ in this
environment. We have that
\[
  \subst{\x:=[\![y]\!]} = \subst{\x:=v'} \quad\less\quad
  \sigma_2\comp\subst{\x:=[\![x]\!]} = \subst{\x:= \CC{A}\CC{A}\C{A}\C{A}v'} \approx \subst{\x:=v'}
  \ .
\]
The ``$\x_2 := \app{\infty}()$'' in the loop for the Ackermann function is
needed because we don't know how to express the second argument at the call
site \ttt{ack m (ack S[m] n)}. The upper arc safely represents this call site
because for all possible values of \ttt{m} and \ttt{n}, the semantics of
\ttt{ack S[m] n} (a natural number) is approximated by~$\app{\infty}()$.
%%%>>>2

\subsection{Collapsing} %%%<<<2

For combinatorial reasons, we will need the labels of the control-flow graph
(substitutions) to live in a finite set. The two main obstructions for the
finiteness of~$\T$ are that the depth of terms is unbounded and that there are
infinitely many possible weights for the approximations~$\app{w}$s.
Define the \emph{constructor depth} and the \emph{destructor depth} of a term
with
\begin{myequation}
\Depth_C\big(\C C t\big)
  &\eqdef&
1+\Depth_C(t)       \Big.            \\
\Depth_C\big((t_1,\dots,t_n)\big)
  &\eqdef&
\max_{1\le i\le n}\big(1+\Depth_C(t_i)\big)    \\
\Depth_C\big(\seq d\big)
  &\eqdef&
0\\
\Depth_C\big(\app{w}\seq d\big)
  &\eqdef&
0
\end{myequation}\label{def:Depth_C}
and the \emph{destructor depth} of simple terms as:
\begin{myequation}
\Depth_D\big(\C C t\big)
  &\eqdef&
\Depth_D(t)         \Big.          \\
\Depth_D\big((t_1,\dots,t_n)\big)
  &\eqdef&
\max_{1\le i\le n}\big(\Depth_D(t_i)\big)    \\
\Depth_D\big(\seq d\big)
  &\eqdef&
|\seq d|\\
\Depth_D\big(\app{w}\seq d\big)
  &\eqdef&
|\seq d|
\ .
\end{myequation}\label{def:Depth_D}
The depth of a sum of simple terms is the maximum of the depth of its
summands. This allows to define the following restriction for terms:
\begin{defi}
  We write~$\T_{D,B}$ for the subset of all~$t\in\T$ s.t.
  \begin{itemize}
    \item $t$ is in normal form
    \item for each~$\app{w}$ appearing in~$t$, we have~$w\in\ZZ_B = \{-B,
      \dots,0,1,\dots, B-1, \infty\}$,
    \item the constructor depth and the destructor depth of~$t$ are less or
    equal than~$D$.
  \end{itemize}
\end{defi}
The aim is to send each element of~$\T$ to an approximation that belongs
to~$\T_{D,B}$.
Given~$B>0$ (fixed once and for all), it is easy to collapse all the weights
in~$\ZZ_\infty$ into the finite set~$\ZZ_B$: send each~$w$ to~$\PB{B}{w}$,
with
\[
  \PB{B}{w}
  \quad\eqdef\quad
  \begin{cases}
    -B     & \hbox{if $w<-B$}\cr
     w     & \hbox{if $-B\leq w<B$}\cr
    \infty & \hbox{if $w\geq B$}
    \ .
  \end{cases}
\]
This gives rise to a function from simple terms to simple terms with bounded
weights: using the grammar of simple terms from Lemma~\ref{lem:normal_form},
we define
\begin{myequation}
  \PB{B}{\C{C}t} &\eqdef& \C{C}{\PB{B}{t}} \\
  \PB{B}{(t_1,..., t_n)} &\eqdef& \big(\PB{B}{t_1}, ..., \PB{B}{t_n}\big)\\
  \PB{B}{\app{w}\seq d\,} &\eqdef& \app{\PB{B}{w}} \, \seq d\\
  \PB{B}{\,\seq d\,} &\eqdef& \seq d
  \ .
\end{myequation}

\bigbreak
Ensuring that the depth is bounded is more subtle. Given~$D\geq0$
(fixed once and for all) and~$t\in\T$ in normal form, we want to bound the
constructor depth and the destructor depth by~$D$. 
This is achieved with the following definition acting on simple
terms and extended by linearity. Because of~$(*)$, the clauses are not
disjoint and only the first appropriate one is used:

\begin{myequation}
\PC{i}{(\C C t)} &\eqdef& \C C (\PC{i-1}{t})
   & \hbox{if $i>0$}\\
\PC{i}{(t_1,\dots,t_n)} &\eqdef&
    \big(\PC{i-1}{t_1}, \dots, \PC{i-1}{t_n}\big)
   & \hbox{if $i>0$}\\
\PC{i}{\big(\app{w}\seq d\big)} &\eqdef& \app{w} \big(\PD{D}{\seq d}\big)
   & \hbox{if $i>0$}\\
\PC{i}{\seq d} &\eqdef& \PD{D}{\seq d}
   & \hbox{if $i>0$}\\
\PC{0}{t} &\eqdef& \PD{D}{\nf(\app{0} t)} & (*)\\
    \noalign{\smallbreak}
    \PD{D}{\big(\app{w}\seq d\big)} &\eqdef& \app{w}\big(\PD{D}{\seq d}\big) &
    (**)\\
\PD{D}{\seq d} &\eqdef& \seq d
   & \hbox{if $|\seq d| \leq D$}\\
\PD{D}{\big(\CC{C}\seq d\big)} &\eqdef& \app{-1}\big(\PD{D}{\seq d}\big)
   & \hbox{if $|\CC{C}\seq d| > D$} \\
\PD{D}{\big(\pi_i\seq d\big)} &\eqdef& \app{-1}\big(\PD{D}{\seq d}\big)
  & \hbox{if $|\pi_i \seq d| > D$} \rlap{\hbox{\ .}}
\end{myequation}
Note that we need to compute a normal form for clause~$(*)$, and that since
the normal form of~$\app{0}t$ doesn't contain any constructor (recall that
approximations absorb constructors on their right),
each summand of the result will match the left side of clause~$(**)$.
%
%Note also that the definition of~$\PC{i}{t}$ is done by double induction on~$i$
%and~$t$ while the definition of~$\PD{D}{\seq d}$ is only by induction on~$\seq
%d$.

\medbreak
The function~$t \mapsto \PC{D}{t}$ does several things:
\begin{itemize}
  \item it keeps the constructors up to depth~$D$ (the first four clauses),
  \item it removes the remaining constructors with~$t\mapsto\app{0} t$ (clause~$(*)$),
  \item it keeps a suffix of at most~$D$ destructors in front of each
    variable and incorporates the additional destructors into the
    preceding~$\app{w}$ (the last three clauses).
\end{itemize}
For example, we have
\[
  \PD{2}{\big(\C{A}\C{B}\C{C}\C{D}\app{w}\CC{X}\CC{Y}\CC{Z}\x\big)}
  \quad=\quad
  \C{A}\C{B}\app{w+1}\CC{Y}\CC{Z}\x
\]
and
\[
  \PD{1}{\big(\C{A}(\x,\C{B}\app{w}\CC{X}\CC{Y}\y)\big)}
  \quad=\quad
  \C{A}\app{1}\x + \C{A}\app{w+1}\CC{Y}\y
  \ .
\]

\begin{lem}\label{lem:collapsing_monotonic}
  The collapsing function~$t \mapsto \PB{B}{\PC{D}{t}}$ is inflationary
  and monotonic:
  \begin{itemize}
    \item $t \less \PB{B}{\PC{D}{t}}$,
    \item if $t \less u$ then~$\PB{B}{\PC{D}{t}} \less \PB{B}{\PC{D}{u}}$,
  \end{itemize}
  More precisely, both functions~$t\mapsto\PB{B}t$ and~$t\mapsto\PC{D}{t}$ are
  inflationary and monotonic.
\end{lem}
\begin{proof}
By definition of~$\less$, $\PB{B}{\BLANK}$ is inflationary and monotonic. It
follows from the fact that~$\lceil\BLANK\rceil : \ZZ_\infty \to \ZZ_B$ is
itself inflationary and monotonic.

That~$\PC{D}{\BLANK}$ is inflationary relies on the fact that~$t\less \app{0}
t$, $\CC C t \less \app{-1} t$ and~$\pi_i t \less \app{-1}t$; it is a direct
inductive proof. The proof that~$\PC{D}{\BLANK}$ is monotonic is a tedious
inductive proof. It is omitted for sake of brevity.

Together, these facts imply that~$\PB{B}{\PC{D}{\BLANK}}$ is both inflationary
and monotonic.
\end{proof}
The next lemma justifies the use of this collapsing function.
\begin{lem}\label{lem:closure_operators}
  For each~$t\in\T$, we have~$\PB{B}{\PC{D}{t}}\in\T_{D,B}$.
  Moreover,~$\PB{B}{\PC{D}{t}}$ is the least term in~$\T_{D,B}$ that
  approximates~$t$. In particular,
  the function~$t \mapsto \PB{B}{\PC{D}{t}}$ is idempotent
  \[
   \PB{B}{\PC{D}{\PB{B}{\PC{D}{t}}}}
   =
   \PB{B}{\PC{D}{t}}
   \ .
  \]
\end{lem}
\begin{proof}
  %We already know that collapsing is inflationary and monotonic.
  It is easy to show that both~$\PB{B}{\BLANK}$ and~$\PC{D}{\BLANK}$ are
  idempotent. Idempotence of~$\PB{B}{\PC{D}{\BLANK}}$ follows from the fact
  that~$\PC{D}{\PB{B}{\PC{D}{t}}} = \PB{B}{\PC{D}{t}}$.
  That~$\PB{B}{\PC{D}{t}}\in\T_{D,B}$ follows directly from the definitions.
  Since it is not needed in this paper, the proof that~$\PB{B}{\PC{D}{t}}$ is
  the least term in~$\T_{D,B}$ that approximates~$t$ is omitted.
\end{proof}
%Since~$\PB{B}{\PC{D}{\BLANK}}$ is also monotonic and inflationary, it is thus
%a \emph{closure operator}.
An interesting corollary of
Lemma~\ref{lem:closure_operators} is that collapsing is monotonic with respect
to the bound~$D$ and~$B$:
\begin{cor}\label{cor:collapsing_antitonic_bounds}
  If~$0\leq D\leq D'$ and~$0<B\leq B'$, then~$\PB{B'}{\PC{D'}{t}} \less
  \PB{B}{\PC{D}{t}}$.
\end{cor}

\begin{defi}\label{def:collapsed_composition}
  If~$\sigma = \subst{\x_1:=t_1;\dots;\x_n:=t_n}$ and~$\tau = \subst{\y_1 :=
  u_1;\dots;\y_m:=u_m}$ are substitutions, then~$\tau\ccomp\sigma$ is defined
  as the pointwise collapsing~$\PB{B}{\PC{D}{(\tau\comp\sigma)}}$.
\end{defi}
This collapsed composition~``$\ccomp$'' is a binary operation on~$\T_{D,B}$.
Unfortunately, it is not associative! For example, when~$B=2$, the composition
\[
  \subst{\ttt{r}:=\app{-1}\x} \ccomp \subst{\x:=\app{1}\y} \ccomp
  \subst{\y:=\app{1}\ttt{z}}
\]
can give~$\subst{\ttt{r}:=\app{1} \ttt{z}}$ or~$\subst{\ttt{r}:=\app{\infty}
\ttt{z}}$ depending on which composition we start with. Similarly, when~$D=1$,
the composition
\[
  \subst{\ttt{r}:=\CC{C}\x} \ccomp \subst{\x:=\C{C}\y} \ccomp \subst{\y:=\C{D}\ttt{z}}
\]
can give~$\subst{\ttt{r}:=\C{D} \ttt{z}}$ or~$\subst{\ttt{r}:=\app{1}
\ttt{z}}$.
There is a special case: when~$D=0$ and~$B=1$, the operation~$\ccomp$ is
becomes associative! This was the original case of SCT~\cite{SCT}. In general,
we have:
\begin{defi}
  Two terms~$u$ and~$v$ are called \emph{compatible}, written~$u\coh v$, if
  there is some~$t\not\approx\Zero$ that is finer than both, i.e., such
  that~$t\less u$ and~$t\less v$. Two
  substitutions are compatible if they are pointwise compatible.
\end{defi}
\begin{lem}\label{lem:compatible_assoc}
  If~$\sigma_1,\dots,\sigma_n$ is a sequence of composable substitutions, and
  if~$\tau_1$ and~$\tau_2$ are the results
  of computing~$\sigma_n\ccomp\dots\ccomp\sigma_1$ in different ways,
  then~$\tau_1 \coh \tau_2$.
\end{lem}
\begin{proof}
  We have~$\sigma_n\comp\dots\comp\sigma_1 \less \tau_1$
  and~$\sigma_n\comp\dots\comp\sigma_1 \less \tau_2$.
\end{proof}
In order to simplify notations, we omit parenthesis and make this operation
associate on the right:~$\sigma_1 \ccomp \sigma_2 \ccomp \sigma_3
\eqdef \sigma_1 \ccomp (\sigma_2 \ccomp \sigma_3)$.
%%%>>>2

%%%>>>1

\section{Size-Change Combinatorial Principle} %%%<<<1

\subsection{Combinatorial Lemma} %%%<<<2

The heart of the criterion is the following combinatorial lemma
\begin{lem} \label{lem:ramsey}
  Let~$G$ be a control-flow graph;
  then, for every infinite path of composable substitutions
  \[
    \ttt{f}_0 \morphism{\sigma_0} \ttt{f}_1 \morphism{\sigma_1} \dots \morphism{\sigma_n} \ttt{f}_{n+1} \morphism{\sigma_{n+1}}\dots
%    \leqno{(*)}
  \]
  in the control-flow graph~$G$, there is a node~$\ttt{f}$ such that the path
  can be decomposed as
  \[
    \ttt{f}_0
    \quad
    \underbrace{\morphism{\sigma_0} \dots \morphism{\sigma_{n_0-1}}}_{\hbox{\scriptsize initial prefix}}
    \quad
    \ttt{f}
    \quad
    \underbrace{\morphism{\sigma_{n_0}} \dots \morphism{\sigma_{n_1-1}}}_{\tau}
    \quad
    \ttt{f}
    \quad
    \underbrace{\morphism{\sigma_{n_1}} \dots \morphism{\sigma_{n_2-1}}}_{\tau}
    \quad
    \ttt{f}
    \quad
    \dots
  \]
  where:
  \begin{itemize}
    \item all the~$\sigma_{n_{k+1}-1}\ccomp\dots \ccomp
      \sigma_{n_k}$ are equal to the same~$\tau:\ttt{f}\to \ttt{f}$,
    \item $\tau$ is \emph{coherent}: $\tau\ccomp\tau\coh\tau$.
  \end{itemize}
\end{lem}
The proof is the same as the original SCT~\cite{SCT}, with only a
slight modification to deal with the fact that~$\ccomp$ isn't associative.
\begin{proof}
  This is a consequence of the infinite Ramsey theorem.
  Let~$(\sigma_n)_{n\geq0}$ be an infinite path as in the lemma. We associate a
  ``color''~$c(m,n)$ to each pair~$(m,n)$ of natural numbers where~$m<n$:
  \[
    c(m,n)
    \quad\eqdef\quad
    \Big(
      \ttt{f}_m \ ,\ %
      \ttt{f}_n \ ,\ %
      \sigma_{n-1} \ccomp\cdots\ccomp \sigma_{m}
    \Big)
    \ .
  \]
  Since the number of constructors and the arity of tuples that can arise from
  compositions in a control flow graph is finite, the number of possible
  colors is finite.
  By the infinite Ramsey theorem, there is an infinite set~$I\subseteq
  \NN$ such all the~$(i,j)$ for~$i<j\in I$ have the same
  color~$(\ttt{f},\ttt{f'},\tau)$.  Write~$I = \{
  n_0<n_1<\cdots<n_k<\cdots\}$. If~$i<j<k \in I$, we have:
  \begin{myequation}
    \Big(
      \ttt{f} \ ,\ %
      \ttt{f}' \ ,\ %
      \tau
    \Big)
  &=&
    \Big(
      \ttt{f}_i \ ,\ %
      \ttt{f}_j \ ,\ %
      \sigma_{j-1} \ccomp\cdots\ccomp \sigma_{i}
    \Big)\\
  &=&
    \Big(
      \ttt{f}_j \ ,\ %
      \ttt{f}_k \ ,\ %
      \sigma_{k-1} \ccomp\cdots\ccomp \sigma_{j}
    \Big)\\
  &=&
    \Big(
      \ttt{f}_i \ ,\ %
      \ttt{f}_k \ ,\ %
      \sigma_{k-1} \ccomp\cdots\ccomp \sigma_{i}
    \Big)
  \end{myequation}
  which implies that~$\ttt{f}=\ttt{f}'=\ttt{f}_i=\ttt{f}_j=\ttt{f}_k$ and
  \begin{myequation}
    \tau
    &=&
    \sigma_{j-1} \ccomp\cdots\ccomp \sigma_{i}\\
    &=&
    \sigma_{k-1} \ccomp\cdots\ccomp \sigma_{j}\\
    &=&
    \sigma_{k-1} \ccomp\cdots\ccomp \sigma_{j}\ccomp
    \sigma_{j-1} \ccomp\cdots\ccomp \sigma_{i}\\
    \tau\ccomp\tau
    &=&
    \big(\sigma_{k-1} \ccomp\cdots\ccomp \sigma_{j}\big)\ccomp
    \big(\sigma_{j-1} \ccomp\cdots\ccomp \sigma_{i}\big)
    \ \hbox{.}
  \end{myequation}
  In the original SCT principle, composition was
  associative and we had~$\tau\ccomp\tau=\tau$. Here
  however,~$\tau$ and~$\tau\ccomp\tau$ differ only in the order of
  compositions, and we only get that~$\tau \coh \tau\ccomp\tau$
  (Lemma~\ref{lem:compatible_assoc}).
\end{proof}

%%%>>>2

\subsection{Graph of paths} %%%<<<2

The graph of paths of a control-flow graph~$G$ is the graph~$G^+$ with the same
vertices as~$G$ and where arcs between~$a$ and~$b$ in~$G^+$ correspond exactly to
paths between~$a$ and~$b$ in~$G$. In our case, the graph is labeled with
substitutions and the label of a path is the composition of the labels of its
arcs.
\begin{defi}
  If $G$ is a control-flow graph, the graph~$G^+$, the \emph{graph of paths
  of~$G$}, is the control-flow graph defined as follows:
  \begin{itemize}
    \item $G^0 = G$,
    \item in $G^{n+1}$, the arcs from~$\ttt{f}$ to $\ttt{g}$ are
    \[
      G^{n+1}(\ttt{f},\ttt{g})
      \quad=\quad
      G^n(\ttt{f},\ttt{g})
      \cup
      \big\{\sigma\ccomp\tau
               \ \mid\ %
            \tau\in G^n(\ttt{f},\ttt{h}),
            \sigma\in G(\ttt{h},\ttt{g})\big\}
    \]
    where~$\ttt{h}$ ranges over all vertices of~$G$,
    \item $G^+ = \bigcup_{n\geq0} G^n$.
  \end{itemize}
\end{defi}
By definition, each path~$\sigma_1\cdots\sigma_n$ in~$G$ corresponds to an
arc in~$G_{D,B}^+$ that is labelled with~$\sigma_n\ccomp\dots\ccomp\sigma_1$.
(Recall that~``$\ccomp$'' associates on the right.)
The restrictions of the sets~$\T_{D,B}(\x_1,\dots,\x_m)$ to terms that can
appear in compositions of arcs in a control-flow graph~$G$ are finite because
the number of variables and constructors in~$G$ is
finite and the arity of tuples is bounded. We thus have:
\begin{lem}
$G^+$ is finite and can be computed in finite time. More
precisely,~$G^n=G^{n+1}$ for some~$n$, and $G^+$ is
equal to this~$G^n$.
\end{lem}
\noindent
As an example, here are the first steps of the computation of the graph of paths
of the control-flow graph for the functions~\ttt{\(\ttt{f}\sb1\)}
and~\ttt{\(\ttt{g}\sb1\)} (page~\pageref{def:f1g1}) when~$D=B=1$. The initial
control-flow graph~$G$ given by the static analysis is given on page~\pageref{graph:f1g1}.
The graph~$G^0 = G$ contains only two arcs:
\begin{itemize}
  \item $\sigma_1 \eqdef \subst{\x:=\C{A}\x}$
    from~$\ttt{f}_1$ to~$\ttt{g}_1$;
  \item $\sigma_2 \eqdef \subst{\x:=\CC{A}\CC{A}\x}$
    from~$\ttt{g}_1$ to~$\ttt{f}_1$.
\end{itemize}
The graph~$G^1$ is then
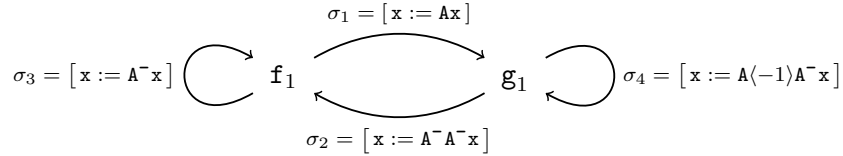
\[
  \begin{tikzpicture}[baseline=(m-1-1.base)]
    \matrix (m) [matrix of math nodes, row sep={3em,between origins},
    column sep={8em,between origins},text depth=.25ex, text height=1.5ex]
      { \ttt{\(\ttt{f}\sb1\)} & \ttt{\(\ttt{g}\sb1\)} \\ };
    \path[edge]

      (m-1-1) edge[loop left, distance=1.5cm, in=150, out=210]
      node[left] {$\sigma_3 = \subst{\x := \CC{A} \x}$} (m-1-1)
      (m-1-2) edge[loop right, distance=1.5cm, in=-30, out=30]
      node[right] {$\sigma_4 = \subst{\x := \C{A}\app{-1}\CC{A} \x}$} (m-1-2)
      (m-1-1) edge[bend left=30] node[auto] {$\sigma_1 = \subst{\x := \C{A} \x}$} (m-1-2)
      (m-1-2) edge[bend left=30] node[auto] {$\sigma_2 = \subst{\x := \CC{A}\CC{A}\x}$} (m-1-1);
  \end{tikzpicture}
\]
where the loop~$\sigma_3$ on the left is obtained as~$\sigma_2 \ccomp
\sigma_1$ and the loop~$\sigma_4$ on the right is obtained
as~$\sigma_1\ccomp\sigma_2$.
The next iteration gives the following arcs for~$G^2$:
\begin{itemize}

  \item $\sigma_5\eqdef\sigma_1\ccomp\sigma_3$ which
    gives~$\subst{\x:=\C{A}\app{-1}\x}$ from~$\ttt{f}_1$ to~$\ttt{g}_1$,

  \item $\sigma_1\ccomp\sigma_2$ which gives~$\sigma_4$ around~$\ttt{g}_1$,

  \item $\sigma_6\eqdef\sigma_2\ccomp\sigma_4$ which
    gives~$[\x:=\app{-1}\CC{A}\x]$ from $\ttt{g}_1$ to~$\ttt{f}_1$,

  \item $\sigma_2\ccomp\sigma_1$ which gives ~$\sigma_3$ around~$\ttt{f}_1$.
\end{itemize}
The next iteration~$G^3$ yields a single new arc: $\subst{\x:=\app{-1}\x}$
from~$\ttt{f}_1$ to~$\ttt{f}_1$. This graph~$G^3$ with 7 arcs is the graph of
paths of the starting control-flow graph.
%%%>>>2

\subsection{Size-Change Termination Principle} %%%<<<2

First, a small lemma:
\begin{lem}\label{lem:values_approx}
  If $v\in\T()$ is an exact value, then:
  \begin{itemize}
    \item the normal form of $\app{0}v$ is of the form~$\sum_i
      \app{w_i}()$ where~$\max_i\big(w_i\big) = \Depth(v)$,
    \item $\app{\Depth(v)}() \less \nf(\app{0}v) \less \app{\Depth(v)}()$,
    \item if $v \less \app{w}()$ then~$w \geq \Depth(v)$.
  \end{itemize}
\end{lem}
\begin{proof}
  The first point is a simple inductive proof on~$v$, and the second point
  follows directly.

  For the third point, suppose that~$v\less \app{w}()$. By contextuality
  of~$\less$ (refer to Definition~\ref{def:preorder}), we get~$\app{0}v \less
  \app{0}\app{w}()\approx\app{w}()$ and so, by the second point and transitivity,
  that~$\app{\Depth(v)}()\less \app{w}()$. We conclude by
  Lemma~\ref{lem:app_preorder}. %
\end{proof}

A subvalue of a value can be accessed by a sequence of destructors. For
example, the right subtree of a binary tree can be accessed
with~$\pi_2\CC{Node}$ in the sense that~$\pi_2\CC{Node}v$ reduces exactly to
the right subtree of~$v$. By the previous lemma, we can get the depth of the
right subtree by precomposing a value with~$\app{0}\pi_2\CC{Node}$.
A \emph{decreasing parameter} is a subvalue of a parameter whose depth decreases
strictly over a given recursive call. For example, in the call
{\small\begin{alltt}
    val rec push_left x =
      match x with Node[\(\ttt{t}\sb1\), Node[\(\ttt{t}\sb2\),\(\ttt{t}\sb3\)]] -> push_left Node[Node[\(\ttt{t}\sb1\),\(\ttt{t}\sb2\)],\(\ttt{t}\sb3\)]
                 | ...
\end{alltt}}\noindent
the right subtree of~$\x$ is decreasing, while neither its left subtree
nor~$t$ itself are decreasing. In the control-flow graph, this call site
becomes a loop labeled with
\[
  \tau \quad=\quad
  \subst{\x :=
    \C{Node}\Big(
      \C{Node}\big(
        \underbrace{\pi_1\CC{Node}\x}_{\ttt{t}_1},
        \underbrace{\pi_1\CC{Node}\pi_2\CC{Node}\x}_{\ttt{t}_2}\big),
      \underbrace{\pi_2\CC{Node}\pi_2\CC{Node}\x}_{\ttt{t}_3}\Big)}
  \ .
\]
Looking at~$\app{0}\pi_2\CC{Node}\x[\tau]$ to get
the depth of the right subtree of the argument after the recursive call, we obtain
\[
  \app{0}\pi_2\CC{Node}\x[\tau]
  \approx
  \app{0}\pi_2\CC{Node}\pi_2\CC{Node}\x
  \quad\less\quad
  \app{-2}\pi_2\CC{Node}\x
  \ .
\]
This means that the depth of the right subtree of the argument~$\x$ has decreased by
(at least) 2 after the recursive call.

\begin{defi}\label{def:decreasing}
  Let~$\tau = \subst{\ttt{x}_1:=t_1;\dots;\ttt{x}_n:=t_n}$ be a loop in a
  control-flow graph. A \emph{decreasing parameter} for~$\tau$ is a branch of
  destructors:~$\xi = \app{0} \seq d\x_i$ such
  that~$\Zero\not\approx\xi[\tau] \less \app{w} \xi$ with~$w<0$ and~$\seq d$ minimal,
  i.e., no strict suffix of~$\seq d$ satisfies the same condition.
  A loop is called \emph{decreasing} when it has a decreasing parameter.
\end{defi}
The minimality condition is purely technical: without it, the loop~$\tau =
\subst{\x:=\C{A}\app{-1}\CC{A}\x}$ would
have~$\xi=\app{0}\CC{X}\CC{A}\x$ as a decreasing argument
because~$\xi[\tau] \approx \app{-2}\CC{A}\x$.
The problem is that~$\C{X}$ has nothing to do with the definition
and~$\CC{X}\CC{A}$ might not even represent a subvalue of the parameter!
A good decreasing parameter would be~$\app{0}\CC{A}\x$. The
minimality condition is necessary to prove the following lemma:
\begin{lem}\label{lem:minimality_decreasing_parameter}
  If~$\xi$ is a decreasing parameter
  for~$\tau$ and~$\Zero \not\approx
  \sigma \less \tau \comp \rho$ then $\Zero \not\approx \xi[\sigma] \less
  \xi [\tau \comp \rho]$ and in particular,~$\xi [\tau \comp
  \rho]\not\approx\Zero$.
\end{lem}
\proof
  Suppose that~$\xi = \app{0}d_1\cdots d_n \x_i$. The
  ``inequality'' follows from monotonicity. The important point is that under
  the hypothesis, we have~$\Zero \not\approx \xi \comp \sigma$.
  The term~$\xi[\sigma]$ is equal
  to~$\app{0}d_1\cdots d_n \sigma(\x_i)$. Suppose by contradiction that
  this reduces to~$\Zero$. Suppose also that~$\sigma$ is in normal form.

  There is only one reduction sequence of~$\app{0}d_1\cdots d_n \sigma(\x_i)$ and for it to
  give~$\Zero$, this reduction sequence needs to use a reduction step
  from group~$(3)$ of Definition~\ref{def:reduction}. In other words, a
  destructor of ~$d_1\cdots d_n$ has to reach an incompatible constructor
  in~$\sigma(\x_i)$.

  Since~$\sigma \less \tau \comp \rho$ by hypothesis, we have
  that~$\sigma(\x_i)\less\tau\comp\rho(\x_i)=\tau(\x_i)\rho$. By
  Lemma~\ref{lem:orderHeadConstructors} we know that all the head constructors
  of~$\tau(\x_i)\rho$ also appear in~$\sigma(\x_i)$.
  It is not difficult to see that the head constructors appearing in the
  normal form of~$\tau(\x_i)$ also appear in~$\tau(\x_i)[\rho]$. This phenomenon
  is general and doesn't depend on~$\tau$ or~$\rho$. It comes from the fact
  that applying a substitution to a term in normal form doesn't interfere with
  its constructors...

  All the constructors of $\tau(\x_i)$ thus appear in~$\sigma(\x_i)$.
  Since~$d_1\cdots d_n$ reaches an incompatible constructor in~$\sigma(\x_i)$, the
  only way for~$d_1\cdots d_n$ to not reach an incompatible constructor
  in~$\tau(\x_i)$ is to reach the end of the constructors
  in~$\tau(\x_i)$ before the end of~$d_1\cdots d_n$. There are two cases:
  \begin{itemize}
    \item either~$d_1\cdots d_n$ reaches an approximation:
    \begin{myequation}
      \app{0}d_1 \cdots d_n \tau(\x_i) & \red & \dots\\
                                       & \vdots &\hbox{$n-k$ reductions}\\
                        & \red   & \app{0}d_1 \cdots d_k \app{w'} \seq{b} \x_i\\
                        & \approx & \app{w'-k}\seq b\x_i\\
                        & \less & \app{w}\seq d \x_i \quad \hbox{with $w<0$}\ .
    \end{myequation}
    By Lemma~\ref{lem:approximation_destructors} we get that~$\seq d$ is a
    suffix of~$\seq b$, and~$w'-k+|\seq d| \leq w + |\seq b|$.
    But then, we have~$\app{0}d_{k+1} \cdots d_n \tau(\x_i) \red^* \app{w'}\seq b
    \x_i$, and we have that~$d_{k+1}\cdots d_n$ is a suffix of~$\seq b$, and~$w'+|\seq
    d_{k+1}\dots d_n| \leq w + |b|$. This implies that the sequence~$d_1\cdots
    d_n$ wasn't minimal as we have~$\app{0}d_{k+1} \cdots d_n
    \tau(\x_i) \less \app{w}d_{k+1}\cdots d_n \x_i$.

    \item The other possibility is that $\seq d$ reaches directly a branch of
      destructors:
    \begin{myequation}
      \app{0}d_1 \cdots d_n \tau(\x_i) & \red & \dots\\
                                       & \vdots & \hbox{$n-k$ reductions}\\
                        & \red   & \app{0}d_1 \cdots d_k \seq{b} \x_i\\
                        & \less & \app{w}\seq d \x_i \quad \hbox{with $w<0$}\ .
    \end{myequation}
    By Lemma~\ref{lem:approximation_destructors},~$\seq d$ is a suffix
    of~$d_1\cdots d_k\cdotp\seq b$ and~$|\seq d| \leq w + |\seq b| + k$. The
    sequence~$d_1\cdots d_n$ isn't minimal because we
    have~$\app{0}d_{k+1}\cdots d_n\tau{\x_i} \approx \app{0}\seq b\x_i$
    with~$d_{k+1}\cdots d_n$ a suffix of~$\seq b$ and~$|d_{k+1}\cdots d_n|
    \leq w+|\seq b|$, i.e.,~$\app{0}d_{k+1} \cdots d_n \tau(\x_i) \less
    \app{w}d_{k+1}\cdots d_n \x_i$.\qed
  \end{itemize}

\noindent We can now state, and prove, the size-change termination principle.
\begin{prop}[Size-Change Termination Principle with Constructors]\label{prop:SCT}
  If~$G$ safely represents some recursive definitions and all coherent
  loops~$\tau\coh\tau\ccomp\tau$
  in~$G^+$ are decreasing, then the evaluation of the functions on
  values cannot produce an infinite sequence of calls.
\end{prop}
\begin{proof}%[Proof of Proposition~\ref{prop:SCT}]

  Suppose the conditions of the proposition are satisfied and suppose that
  function~$\ttt{h}$ on values~$v_1$, \dots, $v_m$ provokes an infinite
  sequence of calls~$c_1\cdots c_n\cdots$.  Write~$\rho_n$ for the arguments
  of call~$c_n$. The~$\rho_n$'s contain first-order values and in
  particular,~$\rho_0$ corresponds to the initial arguments
  of~\ttt{h}:~$\rho_0=\subst{\x_1:=v_1;\dots;\x_m:=v_m}$.
  Let~$\sigma_1\cdots\sigma_n\cdots$ be the substitutions that label the arcs
  of~$G$ corresponding to the calls~$c_1 c_2\cdots$. We can use
  Lemma~\ref{lem:ramsey} to decompose this sequence as:
  \[
    \ttt{h}
    \quad
    \underbrace{\morphism{\sigma_0} \dots \morphism{}}_{\hbox{\scriptsize initial prefix}}
    \quad
    \ttt{f}
    \quad
    \underbrace{\morphism{\sigma_{n_0}} \dots \morphism{}}_{\tau}
    \quad
    \ttt{f}
    \quad
    \underbrace{\morphism{\sigma_{n_1}} \dots \morphism{}}_{\tau}
    \quad
    \ttt{f}
    \quad
    \dots
  \]
  where:
  \begin{itemize}
    \item all the $\sigma_{n_{k+1}-1}\ccomp \dots \ccomp \sigma_{n_k}$
      are equal to the same~$\tau : \ttt{f} \to \ttt{f}$,
    \item $\tau$ is coherent: $\tau \coh \tau\ccomp \tau$.
  \end{itemize}
  The control-flow graph~$G$ is safe and we thus have
  \[
     \rho_{n+1} \quad\less\quad \sigma_n \comp \rho_n
  \]
  Since~$\comp$ is monotonic, we also get
  \[
    \rho_{n_1} \quad\less\quad \sigma_{n_1-1}\comp \cdots \comp\sigma_{n_0}
    \comp \rho_{n_0}
    \ \hbox{.}
  \]
  By associativity of~$\comp$, and because collapsing and composition are
  monotonic, we get
  \[
    \rho_{n_1} \quad\less\quad (\sigma_{n_1-1}\ccomp \cdots \ccomp\sigma_{n_0})
    \comp \rho_{n_0}
    \ =\ \tau\comp\rho_{n_0}
    \ \hbox{.}
  \]
  Repeating this, we obtain:
  \[
    \rho_{n_k}
    \quad\less\quad
    \underbrace{\tau\comp\cdots\comp\tau}_{k} \ \comp\ \rho_{n_0}
    \ .
  \]
  By hypothesis,~$\tau$ has a decreasing parameter: some~$\xi=\app{0}\seq d
  \x$ s.t.~$\xi[\tau]\less\app{w}\xi$ with~$w<0$. We thus have
  \[
    \xi[\rho_{n_k}]
    \quad\less\quad
    \xi[\tau\comp \cdots \comp \tau \comp \rho_{n_0}]
    \quad\less\quad
    \cdots
    \quad\less\quad
    \app{w+\cdots+w}\xi[\rho_{n_0}]
    \ .
  \]
  By Lemma~\ref{lem:minimality_decreasing_parameter}, the right side cannot
  be~$\Zero$. By Lemma~\ref{lem:values_approx}, it is approximated
  by~$\app{w'}()$, where~$w'$ is equal to the depth of the value~$\seq
  d\rho_{n_0}(\x)$.
  We can choose~$k$ large enough to ensure that~$-kw$ is strictly more
  than~$w'$.
  Lemma~\ref{lem:values_approx} also implies
  that~$\xi[\rho_{n_{k}}]$ approximates~$\app{w''}()$, where~$w''$ is
  equal to~$\Depth\big(\seq d\rho_{n_k}(\x)\big)$. But then, we have
  \[
    \app{w''}()
    \quad\less\quad
    \xi[\rho_{n_k}]
    \quad\less\quad
    \app{k w}\xi[\rho_{n_0}]
    \quad\less\quad
    \app{kw + w'}()
  \]
  where~$w''=\Depth(\xi[\rho_{n_k}]) \geq 0$ and~$kw+w'<0$. This contradicts
  Lemma~\ref{lem:values_approx}.
\end{proof}

\medbreak
\begin{defi}
  A control-flow graph~$G$ that satisfies the condition of
  Proposition~\ref{prop:SCT} is said to be \emph{size-change terminating
  for~$D$ and~$B$}.
\end{defi}
We have:
\begin{prop}\label{prop:SCT_monotonicity}
  If~$G$ is size-change terminating for some~$D\geq0$ and~$B>0$, then~$G$ is
  also size-change terminating for all~$D'\geq D$ and~$B'\geq B$.
\end{prop}
\begin{proof}
  Let~$G$ be a control-flow graph, and let~$B'\geq B$ and~$D'\geq D$. Suppose
  that~$G$ is size-change terminating for~$D$ and~$B$; we want to show
  that it is also size-change terminating for~$D'$ and~$B'$.

  Let~$\tau'$ be a coherent loop in~$G_{D',B'}^+$. By construction,~$\tau' =
  \sigma_1\ccomp_{D',B'}\cdots\ccomp_{D',B'}\sigma_n$ for a
  path~$\sigma_1\dots\sigma_n$ in~$G$. We can define~$\tau =
  \sigma_1\ccomp_{D,B}\cdots\ccomp_{D,B}\sigma_n$, which is a loop
  in~$G_{D,B}^+$.

  Since collapsing is monotonic (Lemma~\ref{cor:collapsing_antitonic_bounds}),
  we have that~$\tau'\less\tau$. We also have that~$\tau'\ccomp_{D',B'}\tau'
  \less \tau\ccomp_{D,B}\tau$ and because~$\tau'$ is coherent,~$\tau$ is also
  coherent. By hypothesis,~$\tau$ has a decreasing parameter~$\xi$: we
  have~$\xi[\tau]\less\app{w}\xi$, with~$w<0$. As~$\tau'$: $\xi[\tau'] \less
  \xi[\tau]\less\app{w}\xi$, there is a minimal suffix of~$\xi$ that is a decreasing
  argument for~$\tau'$.
\end{proof}

\subsubsection{The Algorithm} \label{sub:algorithm}

The procedure checking if a set of mutually recursive definitions is
terminating is thus:
\begin{description}

  \item[1- static analysis] compute a safe representation of the recursive
    definitions as a control-flow graph~$G$. The simple static analysis
    described in Appendix~\ref{sec:static_analysis} is enough for all the
    examples in the paper and can be done in linear time.

  \item[2- choose bounds $B$ and $D$] in our implementation, the bounds do not
    depend on~$G$ and are~$B=1$, $D=2$ by default. The user can also change
    them by inserting pragmas together with the code of the recursive
    definitions.

  \item[3- compute the graph of paths] compute the graph of paths~$G^+$
    of~$G$ incrementally, with the bounds~$B$ and~$D$. This step can take an
    exponential amount of space, as the example of the function~\ttt{perms}
    (page~\pageref{ex:perms}) demonstrates.

  \item[4- check coherent loops] check that all the coherent loops of the
    graph~$G^+$ computed  previously are decreasing. If so, the
    functions of the definitions terminate; otherwise, the procedure cannot
    answer.

    For this, it must be possible:
    \begin{itemize}
      \item to check the coherence relation~$\coh$,
      \item to look for decreasing arguments of a loop.
    \end{itemize}
    Some implementation details are given in Appendix~\ref{app:implementation}

\end{description}

\subsubsection*{Failure of Completeness}

The original SCT satisfied a notion of \emph{completeness} stating roughly
that ``all infinite paths are infinitely decreasing \emph{iff} all coherent
loops have a decreasing parameter''. We capture more programs
(Section~\ref{sub:comparison}) than the original SCT, but completeness doesn't
hold anymore. Here is a counter example for~$D=0$ and~$B=2$:
%\footnote{neither this, nor the next example are captured by the original
%SCT.}
{\small\begin{alltt}
  val rec \(\ttt{h}\sb1\) x = match x with A[A[A[x]]] -> \(\ttt{h}\sb2\) x
    and   \(\ttt{h}\sb2\) x = \(\ttt{h}\sb3\) A[X]
    and   \(\ttt{h}\sb3\) x = \(\ttt{h}\sb1\) A[X]
\end{alltt}}\label{cex:completeness}\noindent
The corresponding control-flow graph is
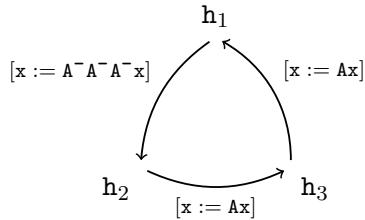
\[\begin{tikzpicture}[scale=1]%change the size here
% nodes
\path (0,1.7320508075688772) coordinate (h1);
\path (-1,0) coordinate (h2);
\path (1,0) coordinate (h3);
% labels
\node [above] at (h1) {$\ttt{h}_1$};
\node [below right] at (h3) {$\ttt{h}_3$};
\node [below left] at (h2) {$\ttt{h}_2$};
% edges
\path[edge] (h1) edge[bend right=25] node[above left]{$[\x:=\CC{A}\CC{A}\CC{A}\x]$} (h2);
\path[edge] (h2) edge[bend right=25] node[below]{$[\x:=\C{A}\x]$} (h3);
\path[edge] (h3) edge[bend right] node[above right]{$[\x:=\C{A}\x]$} (h1);
\end{tikzpicture}
\]
For every conceivable definition of ``decreasing path'', all the infinite
paths in this graph should decrease infinitely. However, because of the
consecutive~``$\subst{\x:=\C{A}\x}$'' arcs, we will get
a~``$\subst{\x:=\app{\infty}\x}$'' arc in the graph of paths, corresponding to
their composition. This will propagate and give coherent
loops~$\subst{\x:=\app{\infty}\x}$ around each node. This graph is not
size-change terminating for~$D=0$ and~$B=2$.

The previous example is size-change terminating whenever~$B>2$; but
completeness doesn't even hold if we can choose the bounds~$B$ and~$D$.
Call a graph~$G$ \emph{decreasing} if no infinite path comes from actual
computation, i.e. if all infinite path evaluate to~$\Zero$. More precisely, it
means that for every infinite
path~$(\sigma_k)_{k>0}$ and substitution~$\rho$ of values, there is a finite
prefix~$\sigma_1\cdots\sigma_n$ s.t.~$\rho\comp\sigma_1\comp\cdots\sigma_n
\approx \Zero$.
The combing function transforming a binary tree into a right-leaning tree
terminates for a subtle reason. Its definition is
{\small\begin{alltt}
  val rec comb x = match x with
       Leaf[] -> Leaf[]
     | Node[t,Leaf[]] -> Node[comb t,Leaf[]]
     | Node[t1,Node[t2,t3]] -> comb Node[Node[t1,t2],t3]
\end{alltt}}\noindent
and it is safely represented by the graph with a single node~\ttt{comb} and
two loops:
\begin{itemize}
  \item $\subst{\x:=\pi_1\CC{Node}\x}$
  \item $\subst{\x:=
      \C {Node}\big(
          \C{Node}(\pi_1\CC{Node}\x,
                   \pi_1\CC{Node}\pi_2\CC{Node}\x),
          \pi_2\CC{Node}\pi_2\CC{Node}\x\big)}$.
\end{itemize}
This graph is terminating in the above sense precisely because~\ttt{comb}
terminates. It can however be shown that for every choice of~$D$ and~$B$, this
graph is \emph{never} size-change terminating.
The reason is that for any bound~$D$ and sequence~$d_1\cdots d_D$ of length~$D$,
there is a tree~$t$ for which the depth of the subtree~$d_1\cdots d_D t$
increases arbitrarily during a sequence of recursive calls. For example,
at~$D=4$ for~$\pi_1\CC {Node}\pi_2\CC {Node}$, consider the tree on the
left:
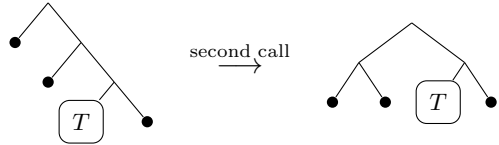
\[
  \begin{tikzpicture}[baseline=(current bounding box.center),
      level distance=15pt,sibling distance=25pt,
      emptynode/.style={circle, inner sep=1.5pt, fill=black},
      treenode/.style={rectangle, rounded corners, draw, inner sep=5pt},
      font=\footnotesize,
      %treenode/.style={rectangle, rounded corners, fill=gray!35, inner sep=5pt},
    ]
    \coordinate
      child { node[emptynode] {} }
      child {
        child { node[emptynode] {} }
        child {
          child { node[treenode] {$T$} }
          child { node[emptynode] {} }
        }
      };
  \end{tikzpicture}
  \quad\stackrel{\hbox{\tiny second call}}{\longrightarrow}\quad
  \begin{tikzpicture}[baseline=(current bounding box.center),
      level 1/.style={level distance=15pt,sibling distance=40pt},
      level 2/.style={level distance=15pt,sibling distance=20pt},
      emptynode/.style={circle, inner sep=1.5pt, fill=black},
      treenode/.style={rectangle, rounded corners, draw, inner sep=5pt},
      font=\footnotesize,
      %treenode/.style={rectangle, rounded corners, fill=gray!35, inner sep=5pt},
    ]
    \coordinate
      child {
        child {node[emptynode] {}}
        child {node[emptynode] {}}
      }
      child {
        child {node[treenode] {$T$}}
        child {node[emptynode] {}}
      }
      ;
  \end{tikzpicture}
  \qquad .
\]
By the second recursive call, the tree on the right will be used as the new
argument. While~$\pi_1\CC {Node}\pi_2\CC {Node}$ corresponds to the empty tree on
the left, it corresponds to~$T$ on the right!
Note that it is the conjunction of the two recursive calls that makes this
possible: for~$\pi_2\CC{Node}\pi_2\CC{Node}$, we need to use the second call
and then the first call:
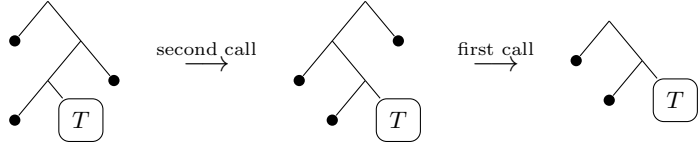
\[
  \begin{tikzpicture}[baseline=(current bounding box.center),
      level distance=15pt,sibling distance=25pt,
      emptynode/.style={circle, inner sep=1.5pt, fill=black},
      treenode/.style={rectangle, rounded corners, draw, inner sep=5pt},
      font=\footnotesize,
      %treenode/.style={rectangle, rounded corners, fill=gray!35, inner sep=5pt},
    ]
    \coordinate
      child {node[emptynode] {}}
      child {
        child {
          child{node[emptynode] {}}
          child{node[treenode] {$T$}}
        }
        child {node[emptynode] {}}
      }
      ;
  \end{tikzpicture}
  \quad\stackrel{\hbox{\tiny second call}}{\longrightarrow}\quad
  \begin{tikzpicture}[baseline=(current bounding box.center),
      level distance=15pt,sibling distance=25pt,
      emptynode/.style={circle, inner sep=1.5pt, fill=black},
      treenode/.style={rectangle, rounded corners, draw, inner sep=5pt},
      font=\footnotesize,
      %treenode/.style={rectangle, rounded corners, fill=gray!35, inner sep=5pt},
    ]
    \coordinate
      child {
        child {node[emptynode] {}}
        child {
          child{node[emptynode] {}}
          child{node[treenode] {$T$}}
        }
      }
      child {node[emptynode] {}}
      ;
  \end{tikzpicture}
  \quad\stackrel{\hbox{\tiny first call}}{\longrightarrow}\quad
  \begin{tikzpicture}[baseline=(current bounding box.center),
      level distance=15pt,sibling distance=25pt,
      emptynode/.style={circle, inner sep=1.5pt, fill=black},
      treenode/.style={rectangle, rounded corners, draw, inner sep=5pt},
      font=\footnotesize,
      %treenode/.style={rectangle, rounded corners, fill=gray!35, inner sep=5pt},
    ]
    \coordinate
      child {node[emptynode] {}}
      child {
        child {node[emptynode] {}}
        child{node[treenode] {$T$}}
      }
      ;
  \end{tikzpicture}
  \quad\hbox{.}
\]
This implies that there can be no decreasing argument in the argument
of~$\ttt{comb}$!

\smallbreak
Surprisingly, adding a second argument representing the size of the tree makes
the function size-change terminating, i.e., the following definition \emph{is}
size-change terminating even though the second argument doesn't decrease at
the second call site.
{\small\begin{alltt}
  val rec comb_size t s = match t,s with
       Leaf[],_ -> Leaf[]
     | Node[t,Leaf[]],S[n] -> Node[comb_size t n,Leaf[]]
     | Node[t1,Node[t2,t3]],n -> comb_size Node[Node[t1,t2],t3],n
     | _,_ -> raise Error[]
\end{alltt}}\noindent
In other words, we can define the combing function as
{\small\begin{alltt}
  val comb t = comb_size t (size t)
\end{alltt}}\noindent
and have the system automatically infer that it is terminating.
%%%>>>2

\subsection{Complexity} %%%<<<2
\label{sub:complexity}

Lee, Jones and ben Amram showed that deciding whether a graph is size-change
terminating in the original sense is P-space hard~\cite{SCT}. We can
encode the same P-space complete problem as an instance of our version of
size-change termination for~$D=0$ and~$B=1$. By monotonicity
(Lemma~\ref{prop:SCT_monotonicity}), all other instances of size-change
termination are P-space hard.

It is not difficult to construct ad-hoc small programs that require an exponential
amount of space, even when~$D=0$ and~$B=1$. The simplest is probably the
following:
{\small\begin{alltt}
    val rec perms \(\x\sb1\) \(\x\sb2\) \(\x\sb3\) \(\x\sb4\) =
      g (perms \(\x\sb2\) \(\x\sb1\) \(\x\sb3\) \(\x\sb4\))
        (perms \(\x\sb1\) \(\x\sb3\) \(\x\sb2\) \(\x\sb4\))
        (perms \(\x\sb1\) \(\x\sb2\) \(\x\sb4\) \(\x\sb3\))
        (perms \(\x\sb4\) \(\x\sb2\) \(\x\sb3\) \(\x\sb1\))
\end{alltt}}\noindent\label{ex:perms}
where \ttt{g} is a previously defined function. The initial control-flow graph
will contain a single node with 4 loops, and the graph of paths will contain 24
loops: one for each permutation of the parameters~$\x_1$ through~$\x_4$. More
generally we can construct, for each~$n$, a program of size~$n^2$ for which
the graph of paths will contain~$n!$ loops. 
However, just like with the original SCT, checking termination of definitions
written by hand with reasonable bounds~$B$ and~$D$ seems to remain practical.
%%%>>>2

\subsection{Comparison with other SCT-Based Criterion} %%%<<<2
\label{sub:comparison}

%\subsubsection*{Original SCT}

In the original SCT, an arc in the control-flow graph was a bipartite graph
with the parameters of the calling function on the left and the arguments of
the called function on the right. A link from~$\x$ to~$u$ can have
label:
\begin{itemize}
  \item $\Down$, meaning that the size of~$\ttt{u}$ is strictly smaller than the
    size of~$\x$,
  \item $\EqDown$, meaning that the size of~$\ttt{u}$ is smaller or equal than
    the size of~$\x$.
\end{itemize}
Such a graph is said to be~\emph{fan-in free} if no~$u$ on the right is the target of
more than one arc. We can encode such a bipartite graph as a
substitution~$\sigma=\subst{\y_1:=t_1;\dots;\y_m:=t_m}$ where:
\begin{itemize}
  \item $t_k = \app{-1} \x_i$ if there is an arc~$\Down$ from~$\x_i$
    to~$u_i$,
  \item $t_k = \app{0} \x_i$ if there is an arc~$\EqDown$ from~$\x_i$
    to~$u_i$,
  \item $t_k = \app{\infty}()$ otherwise.
\end{itemize}
It can be checked that when~$D=0$ and~$B=1$, composition and the size-change
termination condition on~$G^+$ correspond exactly to composition and the
size-change termination condition from~\cite{SCT}. Note in particular that
composition is associative in this context.
Our criterion with~$D=0$ and~$B=1$ is roughly equivalent to the original SCT
for fan-in free graph and with ``depth'' as the notion of size where all arcs
have been initially collapsed. A small lemma stating that checking all
coherent loops~$\tau \coh \tau\ccomp\tau$ is equivalent to checking only the
idempotent loops~$\tau = \tau\ccomp\tau$ is necessary. (A more general
conjecture which we have been unable to prove is that for checking size-change
termination for arbitrary~$D$ and~$B$, it is always sufficient to only check
that idempotent loops have a decreasing argument.)

\subsubsection*{SCT with Difference Constraints}

A. Ben-Amram considered a generalisation of the original SCT which, in our
terminology, could be seen as choosing the bounds~$D=0$ and~$B=\infty$ by
allowing unbounded weights in the control-flow graphs~\cite{deltaSCT}. The
general problem is undecidable, but the restriction to fan-in free graph is
decidable.
The cost of this generality is the introduction of arithmetic in the decision
procedure: deciding if a graph is size-change terminating involves integer
linear programming.
Our control-flow graphs are fan-in free and the criterion avoids arithmetics
by putting a bound on the weights. We lose completeness as shown by the
example on page~\pageref{cex:completeness}, but this doesn't seem to be a
problem in practice because the user may increase the bound~$B$ (at the cost
of speed) and we've rarely found it necessary to go beyond 2 or 3. It would
nevertheless be interesting to see if the approach of~\cite{deltaSCT} can be
combined with our approach to get a criterion for~$B=\infty$ and
arbitrary~$D$.

\subsubsection*{Using ``Calling Contexts''}

P. Manolios and D. Vroon generalized the SCT principle by adding ``calling
contexts'' to the control-flow graph~\cite{calling_contexts}.
A calling context from~\ttt{f} to~\ttt{g} amounts to:
\begin{itemize}
  \item a substitution describing the arguments of~\ttt{g} as terms with free
    variables among the parameters of~\ttt{f},
  \item a set of expressions whose free variables are among the parameters
    of~\ttt{f}.
\end{itemize}
The substitutions are built from the ambient language, as are the expressions
in the set. The intuition of having such a calling context from~\ttt{f}
to~\ttt{g} is that \emph{if all the expressions of the set evaluate to
\ttt{True}, then there can be a call to~\ttt{g} from~\ttt{f}, and the
arguments of~\ttt{g} are given by the substitution}.

This is much more expressive than our approach as the contexts may contain
terms representing arbitrary conditions, like~``$\ttt{Prime(x)}$'' expressing
that a parameter is a prime number. The drawback is that because the
conditions contain free variables, an automatic theorem prover is necessary to
decide when they evaluate to~\ttt{True}. This version has been formalized and
implemented~\cite{Isabelle1,Isabelle2} in Isabelle~\cite{Isabelle}, a proof
assistant based on higher-order logic. The formalization relies Isabelle's
``\ttt{auto}matic'' tactic for checking those conditions.

Our approach uses a similar idea but restricts to the
``constructors/destructors contexts'' that were necessary to build the
arguments of a call. This simplifies the problem so that everything can be
handled combinatorially in a uniform way and makes it more appropriate for a
proof assistant based on type theory like~Coq, or the Agda programming
language.

%%%>>>2

\subsection{Extensions}%%%<<<2

\subsubsection*{Linear Norms}

A lot of attention in the literature on termination has been devoted to
finding a good norm for values~\cite{Lindenstrauss97automatictermination}.
At the moment, the norm used in this paper is very simple: each constructor has
weight~$1$, as can be seen from the reduction~$\app{w}\C C t \red \app{w+1}
t$. Choosing different weights for constructors could be useful in cases such
as
{\small\begin{alltt}
  val rec f = fun
      A[A[A[A[A[B[x]]]]]] ->  f A[A[A[A[A[C[C[x]]]]]]]
    | A[A[A[A[A[C[x]]]]]] ->  f A[A[A[A[A[x]]]]]
    | _ -> A[]
\end{alltt}}\noindent
This function is size-change terminating if the bound~$D$ is greater than~$7$.
If the definition contained other recursive calls, it can make the testing
procedure use more resources than reasonable. Giving a weight of $3$ to $\C
B$ and $1$ to $\C C$ would make this function size-change terminating, even
when~$D=0$. Trying to choose the appropriate weights automatically might not
be worth the trouble but this is still an interesting question.

\subsubsection*{Counting abstractions}
\label{counting_abs}

The PML language for which this criterion was developed is more complete than
the ambient language presented here. In particular, function abstractions and
partially applied functions are allowed. Like OCaml, PML only computes
weak-head normal forms and the function
{\small\begin{alltt} val rec glutton x = glutton
\end{alltt}}\noindent
terminates: when applied to~$n$ arguments, it discards all of them and
stops on the weak-head normal form~\ttt{fun x -> glutton}.

We can make such functions size-change terminating by adding a virtual
parameter~$\x_{ac}$ to all functions. This parameter counts the difference
between the number of abstraction and the number of applications above the
call-site: it gives the ``applicative context'' of the call. This parameter
records an
additional constructor~``$\C{App}$'' introduced by function
application and removed (``$\CC{App}$'') by function abstraction. The
previous function is
{\small\begin{alltt}
  val rec glutton = fun x -> glutton
\end{alltt}}\noindent
which contains an abstraction and no application. The corresponding arc in the
control-flow graph will thus be~$\subst{\ttt{x}_{ac} := \CC{App}\ttt{x}_{ac}}$,
where~$\ttt{x}_{ac}$ is the virtual parameter giving the applicative context
of the call. This virtual parameter makes the definition size-change
terminating (for any choice of bounds~$B$ and~$D$).

This is interesting because dummy abstractions and applications is the
usual way to freeze evaluation and define ``infinite'' data structures in
OCaml.\footnote{Refer to the implementation of the ``\ttt{Lazy}'' module.} In
this context, the size-change termination principle can be used to detect some
notion of \emph{productivity}.
For example, let the type of infinite streams of integers be the
coinductive
type~$\ttt{S} = \ttt{unit -> int*S}$ where~\ttt{unit} is the type with a
single constructor~\ttt{U[]}. The stream of all even integers can be defined with
{\small\begin{alltt}
  val rec arith n d = fun _ -> (n, arith (n+d) d)
  val even = arith Z[] S[S[Z[]]]
\end{alltt}}\noindent
The call~``\ttt{arith n r}'' constructs the stream of integers in arithmetic
progression, starting from \ttt{n} with common difference \ttt{d}.
(``\ttt{\_}'' stands for a dummy variable and~``\ttt{+}'' stands
for the addition of unary natural numbers.) The following definition then corresponds to the
\ttt{map} function on streams:
{\small\begin{alltt}
  val rec map_stream f s = fun _ ->
    match s U[] with
      (n, s) -> (f n, map_stream f s)
\end{alltt}}\noindent
Like~\ttt{glutton}, the functions~\ttt{arith} and~\ttt{map\_stream} have a
deficit of applications: the call-sites are bellow 3~abstractions but only
2~applications. The parameter~$\x_{ac}$ is thus represented
by~$\CC{App}\CC{App}\CC{App}\C{App}\,\C{App}\x_{ac} \approx \CC{App}\x_{ac}$. Those functions are size-change terminating with this
extension: their control-flow graphs consist of a single node with label
\begin{itemize}
  \item $\subst{\ttt{x}_{ac} := \CC{App}\ttt{x}_{ac} ; \ttt{n}
    := \app{\infty}() ; \ttt{d}:=\ttt{d}}$
    for \ttt{arith}
  \item $\subst{\ttt{x}_{ac} := \CC{App}\ttt{x}_{ac} ; \ttt{f}:=\ttt{f} ; \ttt{s} :=
    \app{\infty}()}$ for \ttt{map\_stream}.
\end{itemize}
It is possible to mix finite (inductive) and infinite (coinductive)
structures: here is the function that removes a given number of $0$s in a stream
of integers.
{\small\begin{alltt}
  val rec remove_zeros n s =
    match n with
        Z[] -> s
      | S[m] -> (match s U[] with
                    (Z[], s) -> remove_zeros m s
                  | (S[h], s) -> fun _ -> (S[h], remove_zeros n s))
\end{alltt}}\noindent
The function~\ttt{remove\_zeros} is size-change terminating with this
extension and works for arbitrary streams, i.e., even for those that do not
contain any $0$. Its control-flow graph contains two loops:
\begin{itemize}
  \item
      $\subst{\x_{ac} := \x_{ac};
             \ttt{n} := \CC{S}\ttt{n};
             \ttt{s} := \app{\infty}()}$

  \item
      $\subst{\x_{ac} := \CC{App}\x_{ac};
             \ttt{n} := \ttt{n};
             \ttt{s} := \app{\infty}()}$.

\end{itemize}
A more complete investigation of this phenomenon is pending...

\subsubsection{Higher-Order Arguments}

The PML language for which the size-change termination principle was
implemented allows higher-order arguments for functions. It is not possible to
just ignore higher order arguments: the definition
{\small\begin{alltt}
  val app_zero f = f Z[]
  val rec f x = app_zero f
\end{alltt}}\noindent
might be seen as terminating!

To deal with those, the simplest is to have the static analysis to tag each
instance of a recursively defined function appearing as an \emph{argument} of
another function as non terminating. This makes it possible to define all the
usual functions that have functions in their parameters, like the
real~\ttt{map} function:
{\small\begin{alltt}
  val rec map f x = match x with Nil[]  ->  Nil[]
                               | Cons[a,y]  ->  Cons[f a, map f y]
\end{alltt}}\noindent
whose control-flow graph consists of a single
loop~$\subst{\ttt{f}:=\ttt{f};\x:=\pi_2\CC{Cons}\x}$.
It is also possible to think of smarter static analysis that would see that the
definition
{\small\begin{alltt}
  val phi f = fun n -> math n with Z[] -> Z[]
                                 | S[m] -> n + f m
  val rec f x = phi f x
\end{alltt}}\noindent
is size-change terminating. The PML language uses a constraint checking
algorithm to check that the definitions are well formed, i.e., that their
semantics is well defined. This algorithm builds a kind of \emph{data-flow
graph} to compute an accessibility relation between different parts of the
code and check, for example, that tuples never reach a
``\ttt{match}''~\cite{EJC}. The static analysis is inferred from this
data-flow graph, and it detects that the function~\ttt{phi} defined previously
acts in such a way that~``\ttt{phi f \(u\)}'' may only yield a
call~``$\ttt{f}\CC{S} u$''. Because
of this the control-flow graph of the function~\ttt{f} will contain a single
loop~$\subst{\x:=\CC{S}\x}$, and will thus pass the termination test.

Unfortunately, we currently don't have a proof that this static analysis is
safe! We are currently working on this aspect and are trying to unify the
``data-flow graph'' used for checking that a definition is well-formed and the
``control-flow graph'' used for the SCT.

Finishing the proof that this analysis is safe is interesting because it
allows for a very powerful static analysis. As an example, the following piece of
code is accepted as terminating in the PML language:
{\small\begin{alltt}
  val rec map f l = (* map on lists *)
    match l with
        Nil[] -> Nil[]
      | Cons[a,l] -> Cons[f a, map f l]
  type rec rose_tree A = [  Node[A * list(rose_tree A)] ]
  val rec rmap f t = (* map on rose trees *)
    match t with
        Node[a,l] -> Node[ f a , map (rmap f) l ]
\end{alltt}}\noindent
The data-flow analysis detects that the list~$\ttt{l}$ contains trees that
are smaller than~\ttt{t} and that those elements are fed to the partially
applied~$\ttt{rmap f}$. The control flow-graph for~$\ttt{rmap}$ contains a single
loop~$[\ttt{f}:=\ttt{f}; \ttt{t}:= \pi_1 \CC{Cons}
\pi_2 \CC{Node} \ttt{t}]$. This will be enough for the termination criterion
to accept the function.

%%%>>>1

%%%<<<1 Bibliography
%\bibliographystyle{amsalpha}
%\nocite{*}
%\cite{cousot2011,monotonicityConstraints}???
%mention sized types / size types \cite{Vanhoof02whensize,Barthe,}
\bibliographystyle{amsplain}
\bibliography{sct}
%%%>>>1

\newpage
\appendix

\section{Other Definition of the Approximation Preorder}  %%%<<<1
\label{sec:inductive_approximation}

We give a more concrete characterization of the approximation preorder. This
is crucial in the implementation of the termination test but is also used in
the proof of Lemma~\ref{lem:approximation_destructors}.
The proofs are rather verbose and not very surprising.

\begin{defi}\label{def:inductive_preorder}
  The relation~$\sqless$ is the relation on terms in normal forms generated by
  \[
    \Rule{u \sqless v}%
         {\C{C}u \sqless \C{C}v}%
         {(1)}
     \qquad
     \Rule{u_1 \sqless v_1 \quad\dots\quad u_n\sqless v_n}%
          {(u_1,\dots,u_n) \sqless (v_1,\dots,v_n)}%
          {(2)}
  \]
  \[
  \Rule{\nf\big(\app{0}u\big) \sqless \nf\big(\app{w}v\big)}%
  {u \sqless \nf\big(\app{w}v\big)}%
         {(3)}
  \]
  \[
    \Rule{\forall i=1,\dots,n \  \exists j=1,\dots,m \quad \app{w_i}
    \seq{d_i} \sqless \app{w'_j}\seq{b_j}}%
         {\sum_{i=1}^n  \app{w_i}\seq{d_i}\sqless \sum_{j=1}^m \app{w'_j}\seq{b_j}}%
         {(4)}
  \]
  \[
    \Rule{\hbox{$\seq d$ is a suffix of $\seq b$ and $w'+|\seq d| \leq w+|\seq b|$}}%
         {\app{w'}\seq b\sqless \app{w}\seq d}
         {(5)}
    \qquad
    \Rule{}%
         {\seq d \sqless \seq d}
         {(6)}
  \]
  \emph{where we identify~$\Zero$ and the empty sum.}

  We will usually drop the~$\nf(\BLANK)$ and reason up-to~$\approx$.
\end{defi}

\begin{lem}\label{lem:app_monotonic}
  For every~$u \sqless v$ and~$w\in\ZZ_\infty$, we have~$\app{w}u \sqless
  \app{w}v$.
\end{lem}
\begin{proof}
  By induction on the proof that~$u\sqless v$:
  \begin{itemize}

    \item if the last rule used for~$u\sqless v$ was~$(1)$, then~$u=\C{C}u$
      and~$v=\C{C}v$ and we have~$u\sqless v$. By induction, we have~$\app{w}u
      \sqless \app{w}v$ for all~$w\in\ZZ_\infty$. From this, we conclude
      that~$\app{w}\C{C}u \approx \app{w+1} u \sqless \app{w+1} v \approx
      \app{w} \C{C} v$.

    \item If the last rule was~$(2)$, then~$u=(u_1,\dots,u_n)$
      and~$v=(v_1,\dots,v_n)$ and~$u_i\sqless v_i$ for all~$i=1,\dots,n$. By
      induction, we get~$\app{w}u_i \sqless \app{w}v_i$ for all~$i=1,\dots,n$
      and~$w\in\ZZ_\infty$. We thus have~$\app{w}(u_1,\dots,u_n) \approx
      \sum_i \app{w+1} u_i \sqless \sum_i \app{w+1}v_i \approx
      \app{w}(v_1,\dots,v_n)$.

    \item If the last rule was~$(3)$, then~$v\approx\app{w'}v$ and~$\app{0}u
      \sqless \app{w'}v$. By induction, we get~$\app{w} u \approx
      \app{w}\app{0} u \sqless \app{w}\app{w'} v$ for all~$w\in\ZZ_\infty$.

    \item If the last rule was~$(4)$, then~$u = \sum_{i}
      \app{w_i}\seq{d_i}$ and~$v=\sum_j\app{w'_j}\seq{b_j}$ and
      forall~$i$, there is a~$j$ s.t. $\app{w_i} \seq{d_i} \sqless
      \app{w'_j} \seq{b_j}$. By induction, for all~$i$, there is
      a~$j$ s.t.~$\app{w}\app{w_i} \seq{d_i} \sqless \app{w}\app{w'_j}
      \seq{b_j}$. This implies that~$\sum_{i}
      \app{w}\app{w_i}\seq{d_i} \sqless
      \sum_j\app{w}\app{w'_j}\seq{b_j}$ and because the left side is
      equal to~$\app{w}u$ and the right side is equal to~$\app{w}v$, we
      get~$\app{w}u \sqless \app{w} v$.

    \item If the last rule was~$(5)$, then~$u=\app{w''}\seq b$
      and~$v=\app{w'}\seq d$ with~$\seq d$ a suffix of~$\seq b$ and~$w''+
      |\seq b| \leq w'+|\seq d|$. We have~$\app{w}\app{w''}\seq b
      \approx\app{w+w''}\seq b\sqless\app{w+w'}\seq d\approx
      \app{w}\app{w'}\seq d$ because~$\seq d$ is a suffix of~$\seq b$
      and~$w+w''+
      |\seq b| \leq w+w'+|\seq d|$.

    \item If the last rule was~$(6)$, then~$u=v=\seq d$. We have~$<w>\seq
      d \sqless \app{w}\seq d$ because~$\seq d$ is a suffix of~$\seq d$
      and~$w+|\seq d|\leq w+|\seq d|$.\qedhere

  \end{itemize}
\end{proof}

\begin{lem}\label{lem:sqless_transitive}
  The relation~$\sqless$ is transitive.
\end{lem}
\begin{proof}
  We prove that~$u_1 \sqless u_2$ and~$u_2 \sqless u_3$ implies~$u_1 \sqless
  u_3$ (where each~$u_i$ is in normal form) by induction on the proofs of~$u_2
  \sqless u_3$ and~$u_1 \sqless u_2$. We look at the last rule of~$u_2\sqless
  u_3$:
  \begin{itemize}

    \item If the last rule was~$(1)$, then~$u_2$ is of the form~$\C{C}v_2$,
      and the last rule of~$u_1\sqless u_2$ is necessarily~$(1)$. We thus
      have~$v_1\sqless v_2 \sqless v_3$, which implies by induction,
      that~$v_1\sqless v_3$. This implies that~$\C{C}v_1 \sqless \C{C}v_3$.

    \item If the last rule of~$u_2\sqless u_3$ is~$(2)$, the proof is similar.

    \item If the last rule in~$u_2 \sqless u_3$ was~$(3)$, we have~$u_3 =
      \app{w}v_3$ and~$\app{0}u_2 \sqless \app{w}v_3$. By
      Lemma~\ref{lem:app_monotonic}, we have~$\app{0}u_1\sqless
      \app{0}u_2$, and we know by induction that~$\app{0}u_1 \sqless
      \app{w}v_3$. We get~$u_1 \sqless \app{w}v_3$ by rule~$(3)$.

    \item If the last rule in~$u_2 \sqless u_3$ was~$(4)$, then~$u_2$ is of
      the form~$\sum_j \app{w_{2,j}}\seq{d_{2,j}}$ and~$u_3$ is of
      the form~$\sum_k \app{w_{3,k}}\seq{d_{3,k}}$ and we have~$\forall j,
      \exists k, \app{w_{2,j}}\seq{d_{2,j}}\sqless
      \app{w_{3,k}}\seq{d_{3,k}}$. We look at the last rule of
      the proof that~$u_1 \sqless u_2$:

      \begin{itemize}

        \item if $u_1 \sqless u_2$ ended with~$(3)$, we have~$\app{0} u_1
          \sqless \app{0}u_2\approx u_2 \sqless u_3$. By induction hypothesis,
          $\app{0} \sqless u_3 \approx \app{0}u_3$. We conclude that~$u_1
          \sqless u_3$ by rule~$(3)$.

        \item if~$u_1 \sqless u_2$ ended with~$(4)$, then~$u_1$ is of the
          form~$\sum_i \app{w_{1,i}}\seq{d_{1,i}}$ and for all~$i$, there is
          a~$j$
          s.t.~$\app{w_{1,i}}\seq{d_{1,i}} \sqless
          \app{w_{2,j}}\seq{d_{2,j}}$. By the previous remark about~$u_2
          \sqless u_3$, we thus
          have~$\forall i, \exists k \app{w_{1,i}}\seq{d_{1,i}} \sqless
          \app{w_{3,k}}\seq{d_{3,k}}$. We conclude that~$u_1 \sqless
          u_3$ by rule~$(4)$.

        \item if~$u_1\sqless u_2$ ended with~$(5)$, then~$u_2$ is a sum with a
          single summand and we have~$u_1 \sqless u_2 \sqless \app{w_{3,k}}
          \seq{d_{3,k}}$ for some~$k$. We have~$u_1 \sqless
          \app{w_{3,k}} \seq{d_{3,k}}$ by induction and we get~$u_1
          \sqless \sum_{k} \app{w_{3,k}} \seq{d_{3,k}}$ by
          rule~$(4)$.

      \end{itemize}

    \item If the last rule in~$u_2\sqless u_3$ was~$(5)$, then~$u_3 = \app{w_3}
      \seq{d_3}$ and~$u_2 = \app{w_2}\seq {d_2}$. We look at the last
      rule of the proof that~$u_1\sqless u_2$:

      \begin{itemize}

        \item the proof that~$u_1 \sqless u_2$ ended with~$(3)$. We
          have~$\app{0}u_1 \sqless \app{w_2}\seq{d_2} \sqless
          \app{w_3}\seq{d_3}$, and thus, by induction, that~$\app{0}u_1
          \sqless \app{w_3}\seq{d_3}$. We can use rule~$(3)$ to deduce
          that~$u_1 \sqless \app{w_3}\seq{d_3}$.

        \item the proof that~$u_1\sqless u_2$ ended with rule~$(4)$, with~$u_1 = \sum_i
          \app{w_{1,i}}\seq{d_{1,i}}$. We thus
          have~$\app{w_{1,i}}\seq{d_{1,i}}\sqless \app{w_2} \seq{d_2}$
          for all~$i$s. By induction, we get~$u_{1,i}\sqless
          \app{w_3}\seq{d_3}$ for all~$i$s, and we can conclude that~$\sum_i
          u_{1,i} \sqless \app{w_1}\seq{d_3}$.

        \item the proof that~$u_1\sqless u_2$ ended with rule~$(5)$. We
          have~$u_1 = \app{w_1} \seq{d_1}$ and we get~$\app{w_1}\seq{d_1}
          \sqless \app{w_3}\seq{d_3}$ by rule~$(5)$.

      \end{itemize}

    \item If the last rule of~$u_2\sqless u_3$ is~$(6)$, then the last rule
      of~$u_1\sqless u_2$ is necessarily~$(6)$. We have~$u_1 = u_2 = u_3 =
      \seq d$. Transitivity holds by rule~$(6)$.\qedhere

  \end{itemize}
\end{proof}

\begin{lem}\label{lem:sqless_contextual}
  For every~$u$ in normal form, we have:
  \begin{enumerate}
    \item for every sequence of destructors~$d_1\cdots d_k$, we have~$d_1\cdots
      d_k u \sqless
      \app{-n} u$,
    \item for every sequence of destructors~$d_1\cdots d_k$, if~$u \sqless v$,
      then $d_1\cdots d_k u \sqless d_1\cdots d_k v$,
    \item for every~$t_1\sqless t_2$,~$u[\x:=t_1] \sqless u[\x:=t_2]$.
  \end{enumerate}
\end{lem}

\begin{proof}
  The first point is a simple induction on~$k$:
  \begin{itemize}

    \item if~$k=0]$, the result amounts to~$u\sqless
      \app{0}u$. It follows from rule~$(3)$, $(4)$ and~$(5)$.

    \item if~$\seq d = d_1\cdots d_k\cdotp d_{k+1}$:
      suppose that~$d_{k+1} = \CC{C}$. We look at~$u$:

      \begin{itemize}

        \item if~$u = \C{C}v$, we have~$d_1\cdots d_k \cdotp \CC{C} \C{C}v
          \approx d_1\cdots d_k v \sqless \app{-k} v \approx \app{-k-1} \C{C}
          v$, where the ``inequality'' comes from the induction hypothesis.

        \item if~$u = \C{D} \neq \C{C}$ or~$u=(v_1,\dots,v_n)$, we
          have~$d_1\cdots d_{k+1} u \approx \Zero \sqless \app{-k-1} u$ by rule~$(4)$
          of the definition of~$\sqless$.

        \item if~$u = \seq b\x$ we need to show that~$d_1\cdots d_{k+1}\cdotp \seq b \x
          \sqless \app{-k-1}\seq b \x$. By rule~$(3)$ of the definition
          of~$\sqless$, it is enough to show that~$\app{0}d_1\cdots d_{k+1}\cdotp\seq b
          \x\sqless \app{-k-1}\seq b\x$. This holds by rule~$(5)$.

        \item if~$u=\app{w}\seq b\x$ we need to show~$d_1\cdots d_{k+1}
          \app{w}\seq b \x \approx \app{w-k-1}\seq b \x$ is approximated by~$
          \app{-k-1}\app{w}\seq b\x \approx \app{w-k-1}\seq b \x$. This holds by rule~$(5)$.

      \end{itemize}
      The proof is similar when~$d_{k+1} = \pi_i$.

  \end{itemize}

  \medbreak
  The second point is an induction on the proof that~$t_1 \sqless t_2$.
  If~$k=0$, the result holds trivially. Otherwise, let's assume that the
  sequence of destructors is of the form~$d_1\cdots d_k \CC{C}$.
  \begin{itemize}

    \item If the proof that~$t_1\sqless t_2$ ended with rule~$(1)$, with the
      same constructor~$\C{C}$, we have~$u = \C{C}u'$ and~$v = \C{C}v'$
      with~$u'\sqless v'$. By induction, we can thus conclude that~$d_1\cdots
      d_k\CC{C} u \approx d_1\cdots d_k u' \sqless d_1\cdots d_k v' \approx
      d_1\cdots d_k\CC{C} v$.

    \item If the proof that~$t_1\sqless t_2$ ended with rule~$(1)$ but with a
      different constructor, or with rule~$(2)$, both~$d_1\cdots d_k\CC{C} u$
      and~$d_1\cdots d_k\CC{C} v$ reduce to~$\Zero$, and we conclude with
      rule~$(4)$.

    \item If the proof that~$u\sqless v$ ended with~$(3)$, we know
      that~$v\approx \app{w} v'$ and~$\app{0}u \sqless \app{w}v'$. We
      have~$d_1\cdots d_k\CC{C} \app{0} u \approx \app{-k-1} u \sqless
      d_1\cdots d_k\CC{C} \app{w} v' \approx \app{w-k-1} v'$ by induction. By
      the previous point, we also have~$d_1\cdots d_k\CC{C} u \sqless
      \app{-k-1} u$, and by transitivity, we conclude that~$d_1\cdots
      d_k\CC{C} u \sqless d_1\cdots d_k\CC{C} \app{w} v'$.

    \item If the proof that~$u\sqless v$ ended with~$(4)$, we just need to
      apply the induction hypothesis and rule~$(4)$.

    \item If the proof that~$u\sqless v$ ended with~$(5)$ or~$(6)$, we can conclude
      directly.

  \end{itemize}
  The proof is similar when the sequence of destructors ends with~$\pi_i$.

  \medbreak
  The third point is a simple inductive proof on~$u$:
  \begin{itemize}
    \item if~$u$ is~$\C{C}v$ or~$(v_1,\dots,v_n)$, we just need the induction
      hypothesis and rules~$(1)$ or~$(2)$.
    \item If~$u$ is~$\app{w}\seq d\y$ with~$\y\neq\x$, we just need
      rule~$(5)$. If~$\y=\x$, we need the induction hypothesis and
      Lemma~\ref{lem:app_monotonic}.
    \item If~$u$ is~$\seq d \y$ with~$\y\neq\x$, we just need rule~$(6)$.
      If~$\y=\x$, we use the previous point.\qedhere
  \end{itemize}
\end{proof}
The relation~$\sqless$ isn't quite the same as~$\less$ because it doesn't
interact with~$+$. For example, we don't have~$\x \sqless \x+\y$ or~$\C{C}\x +
\C{D}\x \sqless \app{1}\x$. However, we have:

\begin{lem}\label{lem:inductive_approximation}
  We have~$u\less v$ if and only if~$u$ can be written as~$\sum_i u_i$ and~$v$
  can be written as~$\sum_j v_j$ with~$\forall i,\exists j, u_i\sqless v_j$.
\end{lem}
\begin{proof}

  To show that if~$u\less v$ then~$u$ and~$v$ can be written as sums as in the
  lemma, we define a new relation~$u\sqless' v$ as ``\emph{$u$ can be written
  as~$\sum_i u_i$, $v$ can be written as~$\sum_j v_j$ and~$\forall i, \exists
  j, u_i \sqless v_j$}''. We need to prove that (refer to
  Definition~\ref{def:preorder} page~\pageref{def:preorder}):
  \begin{itemize}
    \item $\sqless'$ is a preorder,
    \item $\sqless'$ is contextual,
    \item $\sqless'$ is compatible with~$\approx$,
    \item $\sqless'$ is compatible with~$+$,
    \item if $w \leq w'$ in~$\ZZ_\infty$, then~$\app{w} t \sqless' \app{w'} t$,
    \item $t \sqless' \app{0}t$.
  \end{itemize}
  Since~$\less$ is the least such relation, we will get that~$u\less v$
  implies~$u\sqless' v$. We only
  sketch the proofs:
  \begin{itemize}
    \item $\sqless'$ is transitive because~$\sqless$ is transitive
      (Lemma~\ref{lem:sqless_transitive}).
    \item $\sqless'$ is reflexive because~$\sqless$ is reflexive on simple
      terms (easy inductive proof).
    \item $\sqless'$ is contextual because~$\sqless$ is contextual
      (Lemma~\ref{lem:sqless_contextual}).
    \item $\sqless'$ is compatible with~$\approx$ because by
      definition,~$\sqless$ is compatible with~$\approx$.
    \item $\sqless'$ is compatible with~$+$ by definition.
    \item That~$\app{w} t \sqless \app{w'} t$ when~$w\leq w'$ in~$\ZZ_\infty$
      is an easy inductive proof. It lifts to~$\sqless'$.
    \item We have~$t\sqless\app{0} t$ by Lemma~\ref{lem:app_monotonic}, and
      this property lifts to~$\sqless'$.
  \end{itemize}

  \medbreak
  The proof that~$\forall i,\exists j, u_i\sqless v_j$ implies that~$\sum_i
  u_i \less \sum_j v_j$ is left as an exercise. It amounts to showing that all
  the rule for~$\sqless$ are valid for~$\less$ and that~$\forall i,\exists j,
  u_i \less v_j$ implies~$\sum_i u_i \less \sum_j v_j$.
\end{proof}
Note that some care is needed to use this lemma to decide approximation on
arbitrary terms. Since~$+$ is associative, commutative and idempotent, there is a
choice to make when writing~$v$ as a sum. For example, we
have~$\C{A}(\x,\y)+\C{B}(\x,\ttt{z}) \less v=\app{2}\x + \app{2}\y +
\app{2}\ttt{z} + \app{1}()$ because we can write~$v$ as~``$\big(\app{2}\x +
\app{2}\y\big) + \big(\app{2}\x + \app{2}\ttt{z}\big) + \dots$'', and we have:
\begin{itemize}
  \item $\C{A}(\x,\y) \sqless \app{2}\x + \app{2}\y$,
  \item $\C{B}(\x,\y) \sqless \app{2}\x + \app{2}\ttt{z}$.
\end{itemize}

%%%>>>1

\section{Implementation Issues}%%%<<<1
\label{app:implementation}

In order to make the presentation readable, the paper followed a rather
abstract description of the criterion. The initial goal was to get a concrete
termination checker for the PML language~\cite{EJC} and ease of implementation
was very important. The code for the criterion can be found at
\url{http://lama.univ-savoie.fr/~hyvernat/Files/basic-SCT.tar.gz}: it consists
of the implementation done for the PML language with a very simple static
analysis for a very simple language. (There are no dependencies for this.) The
full code of PML can be found at \url{http://lama.univ-savoie.fr/~pml/}.

\bigbreak
The main points that make the task relatively straightforward are the
following:
\begin{enumerate}
  \item we only manipulate terms in normal forms and use a representation
    similar to the grammar given in Lemma~\ref{lem:normal_form},
  \item computing if~$t\less u$ and if~$t \coh u$ is easy for those
    terms,
  \item checking if a loop is decreasing (Definition~\ref{def:decreasing}) is
    easy.
\end{enumerate}
Even for terms in normal forms, we need a uniform way to deal with sums.
As $n$-tuples are~$n$-linear, applying linearity to get sums of
simple terms can lead to an exponential blow-up and was ruled out.
We instead start by making sure the initial control-flow
graph doesn't contain any sum. In order to do that, we replace each arc
labeled by a sum with as many arcs as summands. No exponential blow-up occurs
in practice because PML's static analysis doesn't introduce sums. Then,
sums only appear through collapsing of compositions, i.e. from the reduction
rule~$\app{w}(t_1,\dots,t_n) \red \sum_i \app{w+1}t_i$. Those sums
can always be pushed under all constructors and all summands start
with a~$\app{w}$. We thus use the following grammar for terms:
\begin{myequation}\label{gram:implementation}
  t
  &::=&
  \C{C}t              \mskip 15mu|\mskip 15mu
  (t_1,\dots,t_n)     \mskip 15mu|\mskip 15mu
  \seq d              \mskip 15mu|\mskip 15mu
  \textstyle\sum_i \app{w_i} \seq{d_i}\\
  \seq{d}
  &::=&
  \x               \mskip 15mu|\mskip 15mu
  \pi_i \seq{d}    \mskip 15mu|\mskip 15mu
  \CC{C} \seq{d}
\end{myequation}
where the sums are not empty.
Note that~$()$ isn't part of the grammar. It was only used as a presentational
artifact and can be removed from the implementation. Its only concrete use was
to represent an argument whose shape in unknown:~$\app{\infty}()$. In the
implementation, we use~$\sum_{1\leq i\leq a}\app{\infty}\x_i$ instead,
where~$a$ is the arity of the
calling function.

%This can be represented by the following
%inductive type (Caml syntax):
%{\small\begin{alltt}
%  type z_infty = Number of int | Infty
%  type destructor = Project of string | RemoveConstructor of string
%  type parameter = int
%  type term = Constructor of string*argument
%            | Tuple of argument list
%            | Epsilon of (destructor list)*parameter
%            | Sum of (z_infty*(destructor list)*parameter) list
%        \end{alltt}}

\medbreak
All the substitutions are in normal form and composition needs to
do some reduction. This is done using the rules from
Definition~\ref{def:reduction}, with a particular proviso for group~$(3)$:
\begin{itemize}

  \item rules~$\pi_i \C C t \red \Zero$,~$\CC C (t_1,\dots,t_n) \red \Zero$
    and~$\pi_i(t_1,\dots,t_n) \red \Zero$ (when~$i>n$) all raise an
    \emph{error}~\ttt{TypingError}. Encountering such a reduction means that
    the definitions where not valid to begin with and that the initial
    type-checking / constraint solving of the definitions is broken.

  \item the rule~$\CC C \C D t \red \Zero$ raises an
    \emph{exception}~\ttt{ImpossibleCase}. Even safe definitions may introduce
    such reductions, but we know that evaluation will never go along such a path:
    evaluation of~\ttt{match v with ...} may only enter a branch if the
    corresponding pattern matches~\ttt{v}. Compositions raising this exception
    are simply ignored.

\end{itemize}

\subsection{Order, Compatibility and Decreasing Arguments}%%%<<<2

When the terms are generated by the above grammar, we can give an inductive
definition of both the order and the compatibility relation. The inductive
definition of the order corresponds in fact to
Definition~\ref{def:inductive_preorder}:~$\sqless$ is exactly the restriction
of~$\less$ on the terms used in the implementation.

\medbreak
Checking compatibility for arbitrary terms can be subtle. For example, we have
\[
  \big(\app{0} s, (u,v)\big) + \big(\app{0} t, (u,v)\big)
  \quad\coh\quad
  \big((s,t),\app{0} u\big) + \big((s,t),\app{0} v\big)
\]
even though no summand on the left is compatible with a summand on the
right.
However, for the restriction used in the implementation, we can give a purely
inductive definition of compatibility:
\begin{lem}
  Compatibility on terms given by the grammar on
  page~\pageref{gram:implementation} is generated by the following rules:
  \[
    \Rule{u \coh v}%
         {\C{C}u \coh \C{C}v}%
         {(1)}
     \qquad
     \Rule{u_1 \coh v_1 \quad\dots\quad u_n\coh v_n}%
          {(u_1,\dots,u_n) \coh (v_1,\dots,v_n)}%
          {(2)}
  \]
  \[
    \Rule{u \coh \sum_{j=1}^m  \app{w_j}\seq{d_j}}
         {\C{C}u \coh \sum_{j=1}^m  \app{w_j}\seq{d_j}}
         {(3)}
    \qquad
    \hbox{and symmetric}
  \]
  \[
    \Rule{\forall i=1,\dots,n\quad u_i \coh \sum_{j=1}^m \app{w_j}\seq{d_j}}
    {(u_1,\dots,u_n) \coh \sum_{j=1}^m  \app{w_j}\seq{d_j}}
         {(4)}
    \qquad
    \hbox{and symmetric}
  \]
  \[
    \Rule{\exists i=1,\dots,n \  \exists j=1,\dots,m \quad \app{w_i}
    \seq{d_i} \coh \app{w'_j}\seq{b_j}}%
         {\sum_{i=1}^n  \app{w_i}\seq{d_i}\coh \sum_{j=1}^m \app{w'_j}\seq{b_j}}%
         {(4)}
  \]
  \[
    \Rule{\hbox{$\seq d$ is a suffix of $\seq b$ or $\seq b$ is a suffix
    of $\seq d$}}%
         {\app{w'}\seq b\coh \app{w}\seq d}
         {(5)}
    \qquad
    \Rule{}%
         {\seq d \coh \seq d}
         {(6)}
  \]
\end{lem}
Both definitions can be implemented easily using ML pattern matching.
%For example,
%here is the beginning of the Caml definition of compatibility:
%{\small\begin{alltt}
%  let rec compatible t1 t2 = match t1,t2 with
%        Constructor(c1,t1), Constructor(c2,t2) -> c1=c2 && compatible t1 t2
%      | Tuple(l1), Tuple(l2) ->
%          List.length l1 = List.length l2 && List.for_all2 compatible l1 l2
%      | Epsilon(ds,x), Epsilon(es,y) -> ds=es && x=y
%      | Constructor(_,t), Sum(s)
%      | Sum(s), Constructor(_,t) ->
%          let s_minus = List.map (fun (w, ds, x)->(decr w, ds, x)) in
%          compatible t (Sum(s_minus))
%      ...
%      ...
%      | _,_ -> false
%  \end{alltt}}
%  \TODO{...}

\medbreak
Looking for decreasing arguments in a substitution is simple: the minimality
condition means that a decreasing argument is a subterm of one component of
the substitution. It is thus enough to check all subterms!

%%%>>>2

\subsection{Complexity}  %%%<<<2

We saw in section~\ref{sub:complexity} that the problem of deciding
size-change termination is P-space hard. In practice, we have found the
algorithm described on page~\pageref{sub:algorithm} to perform quite well. In
our experience, we have found that checking termination of functions written
by hand in PML doesn't require too much resources. 
There are concrete examples where~$D$ needs to be more than~$4$, but choosing
a bound~$B$ greater than~$1$ is very rarely necessary. The default is to
have~$D=2$ and~$B=1$, and let the user change the bounds. With this default,
termination checking is an order of magnitude faster than sanity checking of
the definitions, except for those examples specifically designed to stress the
system.

There are however two points that help make the criterion perform well,
especially when the bounds~$B$ and~$D$ are greater than their default values:
\begin{itemize}

  \item we make sure that sums are minimal by keeping only maximal
    summands:~$\seq d$ is equivalent (``$\less$ and~$\more$'') to~$\seq d +
    \app{12} \seq d + \app{-1} \CC C \seq d$ and is a much better choice
    because it keeps the size of the graph smaller.

  \item since everything is monotonic with respect to~$\less$, we don't need
    to keep arcs that are approximated by another arc (``subsumption'').

\end{itemize}
These points are trivial to implement and lower the complexity of the
algorithm in practice.

%%%>>>2
%%%>>>1

\section{Static Analysis} %%%<<<1
\label{sec:static_analysis}

%The easiest static analysis interprets each call
%by~$[\x_1:=\app{\infty}();\dots;\x_n:=\app{\infty}()]$. This is safe by
%definition, but unsurprisingly, no resulting control-flow graph is
%size-change terminating.
%
The simplest interesting static analysis only records pattern matching and
projection: for each call-site~``\ttt{g \(u\sb1\) ... \(u\sb{m}\)}'' in the definition
of~``\ttt{f \(\x\sb1\) ... \(\x\sb{n}\)}'', we construct the
substitution~$[\y_1:=u_1;\dots;\y_m:=u_m]$ where
each~$u_i\in\T(\x_1,\dots,\x_n)$ is
\begin{itemize}
  \item a simple term without~$\app{w}$s if~\ttt{\(u\sb{i}\)} is syntactically built from
    projections and pattern-matching variable coming from~$\x_1$, \dots,
    $\x_n$;
  \item $\app{\infty}()$ otherwise.
\end{itemize}
For example, all the examples \ttt{map}, \ttt{\(\ttt{f}\sb1\)}, \ttt{\(\ttt{g}\sb1\)}, \ttt{\(\ttt{f}\sb2\)} and
\ttt{push\_left} (page~\pageref{map_list} and~\pageref{def:f1g1}) yield substitutions
without~$\app{\infty}()$. For the~\ttt{ack} function however, the three
recursive calls are represented by:
\begin{itemize}
  \item $[\x_1:=\CC S \x_1;\C S \C Z ()]$,
  \item $[\x_1:=\C S \CC S \x_1;\x_2:= \CC S \x_2]$,
  \item $[\x_1:=\CC S \x_1 ;\x_2:=\app{\infty}()]$.
\end{itemize}
The~``$\app{\infty}()$'' comes from the call~``\ttt{ack m (ack ...)}'':
because the second argument is an application, it isn't syntactically built
from the parameters. Note that this doesn't prevent the criterion from
tagging the~\ttt{ack} function as terminating.

It should be noted that this static analysis is entirely syntactical and can
be done in linear time in the size of the recursive definitions. This is
similar to the static analysis done in
\url{http://lama.univ-savoie.fr/~hyvernat/Files/basic-SCT.tar.gz}. The only
differences are that:
\begin{itemize}
  \item the syntax of the definitions is much simpler,
  \item we use~$\app{\infty}\x_1 + \cdots + \app{\infty}\x_a$ instead
    of~$\app{\infty}()$.
\end{itemize}

\pdfinfo{
  /Title (The Size-Change Termination Principle for Constructor Based Languages)
  /Author (Pierre Hyvernat)
  /Subject (Computer Science)
  /Keywords (program analysis, termination analysis, size-change principle, ML)}

\end{document}